\documentclass[11pt]{article}
\usepackage{arxiv}

\usepackage[utf8]{inputenc} 
\usepackage[T1]{fontenc}    
\usepackage{hyperref}       
\usepackage{url}            
\usepackage{booktabs}       
\usepackage{amsfonts}       
\usepackage{nicefrac}       
\usepackage{microtype}      
\usepackage{lipsum}		
\usepackage{graphicx}
\usepackage{algorithm}
\usepackage{algorithmic}
\usepackage{authblk}
\usepackage{natbib}
\usepackage{times}
\usepackage{soul}
\usepackage{cancel}

\usepackage{url}
\usepackage{graphicx}
\usepackage{amsmath}
\usepackage{amsthm}
\usepackage{booktabs}
\usepackage{algorithm}
\usepackage{algorithmic}
\usepackage{amsfonts}       
\usepackage{nicefrac}       
\usepackage{microtype}      
\usepackage{xcolor}         
\usepackage{caption}
\usepackage{algorithm}
\usepackage{algorithmic}
\usepackage{graphicx}
\usepackage{subcaption}
\usepackage{amssymb}
\usepackage{mathtools}
\usepackage{xr}
\usepackage{threeparttable}
\newtheorem{example}{Example}
\newtheorem{theorem}{Theorem}
\newtheorem{assumption}{Assumption}
\newtheorem{convassumption}{Convergence Assumption}
\newtheorem*{assumption*}{Assumption}
\newtheorem{definition}{Definition}
\newtheorem{proposition}{Proposition}
\newtheorem{lemma}{Lemma}
\newtheorem{remark}{Remark}
\newtheorem*{claim*}{Claim}

\DeclareFontFamily{U}{matha}{\hyphenchar\font45}
\DeclareFontShape{U}{matha}{m}{n}{
	<5> <6> <7> <8> <9> <10> gen * matha
	<10.95> matha10 <12> <14.4> <17.28> <20.74> <24.88> matha12
}{}
\DeclareSymbolFont{matha}{U}{matha}{m}{n}
\DeclareFontSubstitution{U}{matha}{m}{n}

\DeclareFontFamily{U}{mathx}{\hyphenchar\font45}
\DeclareFontShape{U}{mathx}{m}{n}{
	<5> <6> <7> <8> <9> <10>
	<10.95> <12> <14.4> <17.28> <20.74> <24.88>
	mathx10
}{}
\DeclareSymbolFont{mathx}{U}{mathx}{m}{n}
\DeclareFontSubstitution{U}{mathx}{m}{n}

\DeclareMathDelimiter{\vvvert}{0}{matha}{"7E}{mathx}{"17}

\title{The Causal Impact of Credit Lines on Spending Distributions}

\author[1]{Yijun Li$^*$}
\author[1]{Cheuk Hang Leung$^*$}
\author[2]{Xiangqian Sun}
\author[1]{Chaoqun Wang}
\author[1]{Yiyan Huang}
\author[3]{Xing Yan}
\author[1]{Qi Wu $^\mathcal{y}$}
\author[4]{Dongdong Wang}
\author[4]{Zhixiang Huang}

\affil[1]{School of Data Science, City University of Hong Kong}
\affil[2]{Department of Financial and Actuarial Mathematics, Xi’an Jiaotong Liverpool University}
\affil[3]{Institute of Statistics and Big Data, Renmin University of China}
\affil[4]{JD Digits}

\date{}

\renewcommand{\headeright}{}
\renewcommand{\undertitle}{}
\renewcommand{\shorttitle}{38th Annual AAAI Conference on Artificial Intelligence, Vancouver, Canada.}

\hypersetup{
pdftitle={A template for the arxiv style},
pdfsubject={q-bio.NC, q-bio.QM},
pdfauthor={David S.~Hippocampus, Elias D.~Striatum},
pdfkeywords={First keyword, Second keyword, More},
}

\begin{document}
\maketitle
\def\thefootnote{*}\footnotetext{These authors contributed equally to this work.}
\def\thefootnote{$\mathcal{y}$}\footnotetext{The corresponding author (qiwu55@cityu.edu.hk).}

\begin{abstract}
	Consumer credit services offered by e-commerce platforms provide customers with convenient loan access during shopping and have the potential to stimulate sales. To understand the causal impact of credit lines on spending, previous studies have employed causal estimators, based on direct regression (DR), inverse propensity weighting (IPW), and double machine learning (DML) to estimate the treatment effect. However, these estimators do not consider the notion that an individual's spending can be understood and represented as a distribution, which captures the range and pattern of amounts spent across different orders. By disregarding the outcome as a distribution, valuable insights embedded within the outcome distribution might be overlooked. This paper develops a distribution-valued estimator framework that extends existing real-valued DR-, IPW-, and DML-based estimators to distribution-valued estimators within Rubin’s causal framework. We establish their consistency and apply them to a real dataset from a large e-commerce platform. Our findings reveal that credit lines positively influence spending across all quantiles; however, as credit lines increase, consumers allocate more to luxuries (higher quantiles) than necessities (lower quantiles). Our code is available at \url{https://github.com/lyjsilence/The-Causal-Impact-of-Credit-Lines-on-Spending-Distributions}.
\end{abstract}

\section{Introduction}
``Buy now, pay later'' (BNPL) is a FinTech credit product offered by e-commerce platforms that allow consumers to make purchases first and defer payments later. BNPL is becoming increasingly popular due to its convenience in online shopping \citep{guttman2023buy}.  In practice, e-commerce platforms assign different credit lines (the total amount of money that the platforms lends to a consumer) to potential customers according to their personal information and the history of purchases,  payments, and default behaviors. 

The primary goal of e-commerce platforms in introducing BNPL is to alter the consumption behavior of consumers, which is usually characterized as a specific \textbf{\textit{spending distribution}} formed by the consumption amounts of the consumer’s all orders. The spending distributions of various consumers are different. For example, in Figure \ref{fig:spending distribution}, the spending distribution of some consumers may exhibit a long tail, indicating a preference for both low-price necessities and high-price luxury items, whereas other consumers focus more on middle-valued products.

\begin{figure}
	\centering
	\includegraphics[width=0.5\columnwidth]{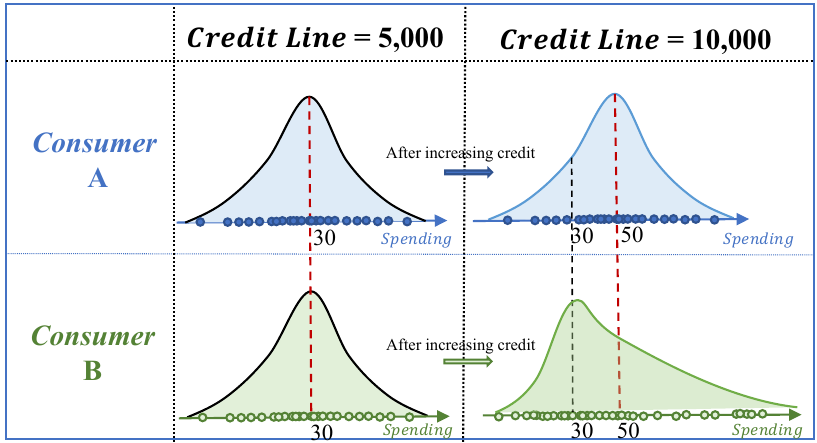}
	\caption{An example for the impact of credit lines change on spending distribution shift (one point stands for spending of one order). \label{fig:example}}
\end{figure}

An essential question for e-commerce platforms is whether and how credit lines affect the consumption behavior of consumers.  Previous studies have shown that increasing credit lines can lead to increased spending amounts, e.g., \cite{aydin2022consumption, soman2002effect}. Nevertheless, they use a scalar quantity (e.g., \textbf{\textit{average spending}} of all orders) to represent the spending of each consumer, which overlooks the complexity of consumption behaviors. For example, consider two consumers (A and B) in Figure \ref{fig:example}. When the credit lines of them both equal 5,000, their spending distributions formed by 50 orders are the same, with an average spending of 30 dollars. Supposing the platform increases their credit lines to 10,000, consumer A prefers to increase the spending of all the orders by 20 dollars, and thus the shape of spending distribution does not change but parallelly shifts to the right by 20. On the other hand, consumer B prefers to purchase more luxury goods and remains the spending amounts of orders for necessities unchanged. The shape of consumer B's spending distribution has shifted dramatically, but the average spending is the same as the first consumer (also increased from 30 to 50).  Even though these two consumers have the same average spending, their spending behaviors are distinct after the change in credit lines. In this case, focusing only on the average spending loses some of the information of distribution (e.g., the part of quantile information). To this end, we propose to investigate \textit{how the changes of credit lines affect the shift of spending distributions}.  However, this raises another question:  since classical causal inference literature targets the outcome of each individual as a scalar, \textit{how can we perform causal inference when the outcome of each individual is a distribution?}

In this paper, we employ a novel causal framework to tackle this problem, where the outcome of each unit is a distribution, and the treatment takes multiple values.  Based on Rubin's causal framework \citep{rubin1977assignment, rubin1978bayesian, rubin2005causal}, we propose three estimators of target quantities: Direct Regression (DR) estimator, Inverse Propensity Weighting (IPW) estimator, and Doubly Machine Learning (DML) estimator. We first study the statistical asymptotic properties of these estimators. Then, to implement these estimators, we develop a deep-learning-based model named \textbf{N}eural \textbf{F}unctional \textbf{R}egression \textbf{Net} (\textbf{NFR Net}) to estimate the complex relationship between functional output and scalar input.  To assess the effectiveness of our methods, we conduct a simulation study. The results reveal that all three estimators are effective, especially for the DML estimator. We finally apply our approach to investigate the causal impact of credit lines on spending distributions based on a real-world dataset collected from a large e-commerce platform. We find that when credit lines increase, consumers' spending tends to rise, which aligns with previous literature. Additionally, we reveal that the impact of credit lines is more significant in the high-quantile range of spending distribution, suggesting that the increase in credit lines is associated with greater demands for luxury goods rather than necessities.

Our contributions can be summarized as follows:
\begin{itemize}
	\item This is the first paper that explores the causal impact of credit lines on spending when the spending of each consumer is summarized as a distribution. Compared to the literature, we discover more detailed findings on the distribution quantiles.
	\item We consider the treatment takes multiple values, and we propose three estimators (i.e., DR, IPW, and DML estimators) for the target quantities. We study the statistical properties of these three estimators and compare them through a simulation experiment. 
	\item The relation between functional output and scalar input is always non-linear and complex. Compared to existing works that captured it by a parameterized function or a linear function, we develop a deep learning model named NFR Net to learn this relationship.
\end{itemize}

\section{Related Work}
Causal inference is a significant challenge in various fields, such as finance \citep{huang2021causal} and health care \citep{shi2019adapting}. 
The key assumption of classical causal inference is that, given the treatment $D=d$, all the units have the same potential outcome distribution (unconditional). As a result, the realization of the outcome for each individual is a \textit{scalar point} drawing from that potential outcome distribution (for instance, in Figure \ref{fig:ATE and QTE} when $D=d$, the blue (red) point is a realization of the $i^{\text{th}}$ ($j^{\text{th}}$) unit). Under the assumption, several causal quantities are introduced and studied. For instance, the average treatment effect (ATE) \citep{chernozhukov2018double} is the difference between the means of any two potential outcome distributions (i.e., $\mathbb{E}[Y(D=\bar{d})]-\mathbb{E}[Y(D=d)]$, or see the left half of Figure \ref{fig:ATE and QTE}). Another quantity is the quantile treatment effect (QTE) \citep{chernozhukov2005iv} that studies the difference between two potential outcome distributions at $\tau$-quantiles (i.e., $Q(\tau, Y(D=\bar{d}))-Q(\tau, Y(D=d))$), or see the right half of Figure \ref{fig:ATE and QTE}).

Various methods have been proposed to estimate the causal effect between treatment and outcome. 
A common approach is constructing the estimators for the target quantities. For example,  Direct Regression (DR) incorporates all confounding factors into a single regression function. The inverse propensity weighting (IPW) method \citep{rosenbaum1983central, hirano2003efficient}, on the other hand, assigns weights to the units based on their propensity scores which mimic RCTs in the pseudo population. However, both of them require accurate estimations of the nuisance parameters, such as the regression function and propensity scores. Doubly Machine Learning (DML) \citep{chernozhukov2018double} method overcomes the shortcomings. It has the doubly robust property such that the accuracy of estimating nuisance parameters can be loosened. 

The above methods are restricted when the outcome of each unit includes many observations or points and they constitute a \textit{distribution}. For example, the shopping spending of a consumer may differ each time, and all the spending amounts form a distribution. In this case, it is impossible to infer the causal relationship via the standard framework unless we reduce the distributions to points (e.g., take the mean).  Thus, it is necessary to seek alternative frameworks for distributional outcomes.

The distributional outcome can be treated as a continuous function. It is closely related to the field of functional data analysis that analyzes data under information varying over a continuum \citep{ramsay2005fitting, wang2016functional, cai2022robust, chen2016variable}. \cite{jacobi2016bayesian} and \cite{chib2007modeling} apply the functional data analysis to study the relationship between functional outcomes and independent variables based on the panel dataset. Nevertheless, they do not focus on the causal studying. \cite{ecker2023causal} considers a causal framework to study the impact of treatment on the functional outcome. However, their work conducts causal inference in Euclidean space. It is believed that the random structure of the distributional outcome is destroyed in the Euclidean space \citep{verdinelli2019hybrid, panaretos2019statistical}. As such, \cite{lin2023causal} considers the causal study in the Wasserstein space, and we extend their framework to study the causal effect on distributional outcomes under multiple treatments and with a deep learning model NFR Net (statistical properties can be ensured as well). In this case, the realization of the outcome for each unit is a \textit{distribution} (for example, in Figure \ref{fig:causal effect map} when $D=d$, the blue (red) distribution is a realization of the $i^{\text{th}}$ ($j^{\text{th}}$) unit).

\begin{figure*}[tbh]
	\begin{minipage}{.53\textwidth}
		\centering
		\includegraphics[width=\columnwidth]{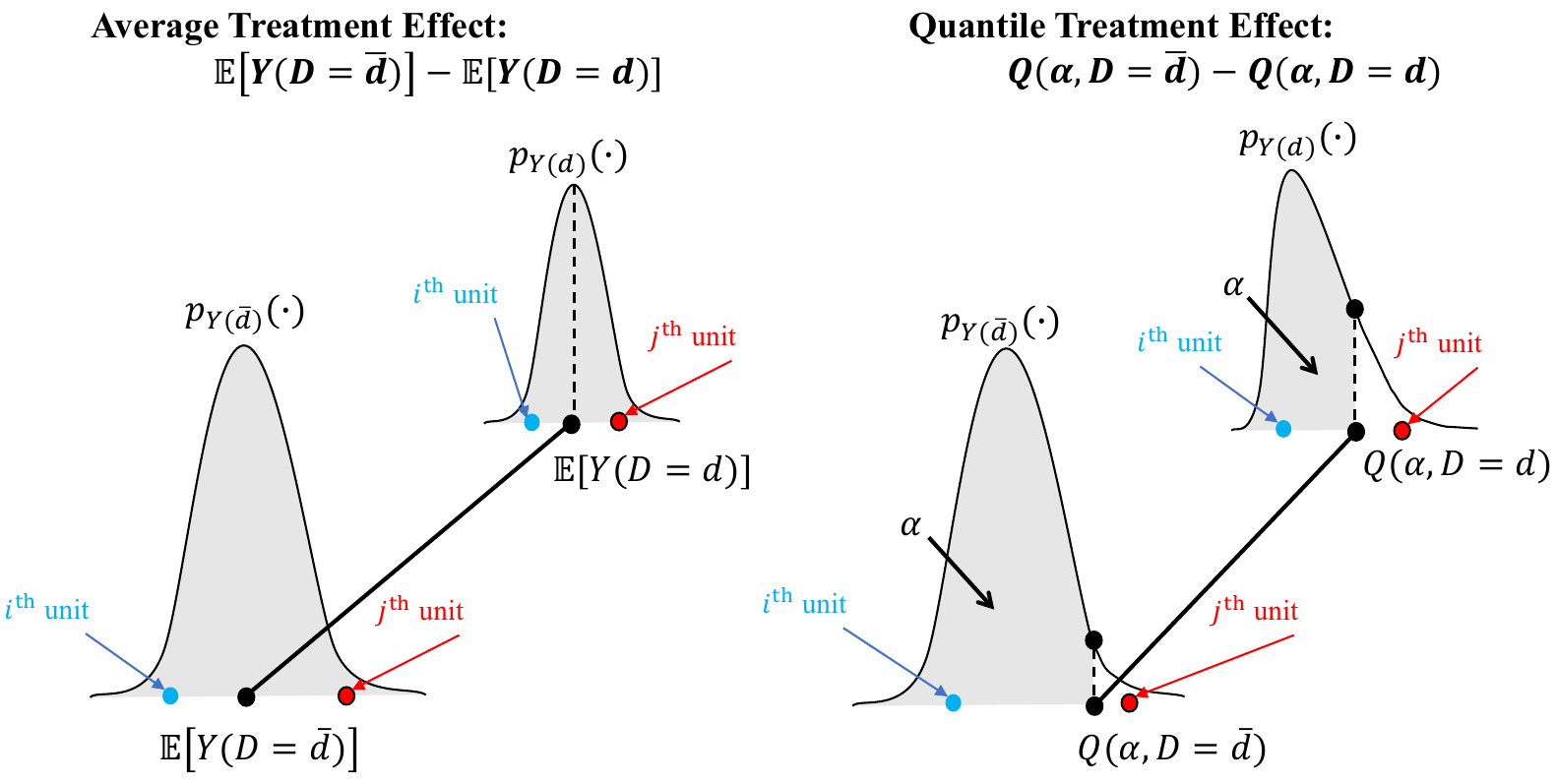}
		\caption{ATE and QTE in the literature. \label{fig:ATE and QTE}}
		\label{fig:high_order_dependencies}
	\end{minipage}
	\quad
	\begin{minipage}{.45\textwidth}
		\centering
		\includegraphics[width=\columnwidth]{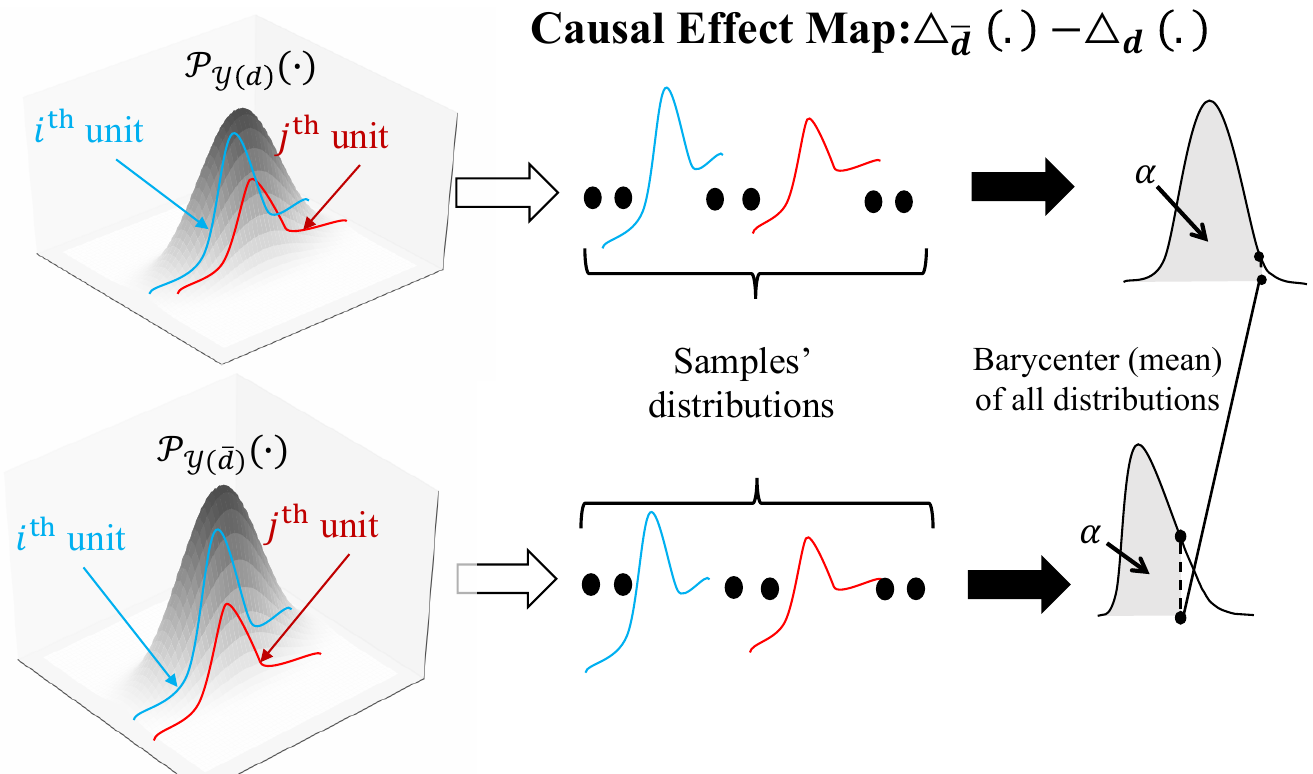}
		\caption{Causal Effect Map in our paper. \label{fig:causal effect map}}
		\label{fig:delinquency distribution}
	\end{minipage}
\end{figure*}

\section{Causal Inference Framework}\label{sec:background}
In this paper, we denote the \textit{treatment} as $D$ such that $D\in\mathfrak{D}=\{d^{1},\dots,d^{r}\}$, the \textit{covariates/confounders} as $\mathbf{X}=[X^{1},\cdots,X^{n}] \in \mathcal{X}$, where $\mathcal{X}$ is a bounded set in $\mathbb{R}^{n}$ with distribution $F_{\mathbf{X}}$.  With scalar outcomes, prior literature defines $Y$ as the outcome variable and $Y(d)$ as the potential outcome variable when receiving treatment $D=d$. Note that $Y=\sum_{i=1}^{r}Y(d^i) \cdot \mathbf{1}_{\{D=d^i\}}$. Accordingly, the \textit{potential outcome distribution and density} of ($Y$, $Y(d)$) is ($F_{Y}$, $F_{Y(d)}$) and ($P_Y, P_{Y(d)}$). In our framework, we consider the case where the outcome of each unit is a distribution that varies across units. To distinguish the differences, we use $\mathcal{Y}$ as the outcome variable and $\mathcal{Y}(d)$ as the potential outcome variable when receiving treatment $D=d$. Similarly, $\mathcal{Y}=\sum_{i=1}^{r}\mathcal{Y}(d^i) \cdot \mathbf{1}_{\{D=d^i\}}$. We can then define the \textit{potential outcome distribution and density} of ($\mathcal{Y}, \mathcal{Y}(d)$) as ($\mathcal{F}_{\mathcal{Y}}$, $\mathcal{F}_{\mathcal{Y}(d)}$) and $(\mathcal{P}_\mathcal{Y}, \mathcal{P}_{\mathcal{Y}(d)})$. We assume that there are $N$-independent units, i.e., $\{(D_{s},\mathbf{X}_{s},\mathcal{Y}_{s})\}_{s=1}^N$.

\subsection{Causal Assumptions}\label{sec:causal framework}
The following causal assumptions are standard under Rubin's framework \citep{rubin2005causal}: (1) \textit{Consistency} (i.e., if $D=d^{i}$ occurs, then $\mathcal{Y}=\mathcal{Y}(d^{i})$ a.s.); (2) \textit{Ignorability/Unconfoundness} (i.e., $\mathcal{Y}(d^{i})\perp\!\!\!\!\perp D|\mathbf{X}, \forall i\in\{1,\dots,r\}$); (3) \textit{Overlap} (i.e., $\mathbb{P}\{D=d^{i}|\mathbf{X}\}>0, \forall i\in\{1,\dots,r\}$). We defer detailed explanations about the causal assumptions and the essentialness of each assumption to Appendix A.

\begin{table*}
	\small
	\caption{Comparisons between our framework and the framework given in the literature.}
	\label{table:Comparisons}
	\centering
	\begin{tabular}{cccccccccccccccccccccc}
		\toprule
		& Our framework     &  Literature framework \\
		\midrule
		Treatment/Covariates variable & $D$/$\mathbf{X}$  &  $D$/$\mathbf{X}$ \\
		Outcome variable   & $\mathcal{Y}$, $\mathcal{Y}(d)$  & $Y$, $Y(d)$ \\
		Potential outcomes distribution (density)  & $\mathcal{F}_{\mathcal{Y}(d)}(\cdot)$ ($\mathcal{P}_{\mathcal{Y}(d)}(\cdot)$) & $F_{Y(d)}(\cdot)$ ($P_{\mathcal{Y}(d)}(\cdot)$)\\
		Metric & Wasserstein & Euclidean\\
		Space of outcome variable  & $\mathcal{W}_{2}(\mathcal{I})$    & $\mathcal{I}\in\mathbb{R}$  \\
		Realization of outcome variable  & distribution       & scalar  \\
		Target quantity  & $\bigtriangleup_{d^{i}}$, $\bigtriangleup_{d^{ij}}$  & $\mathbb{E}[Y(d^{i})]$, $\mathbb{E}[Y(d^{i})]-\mathbb{E}[Y(d^{j})]$ \\
		
		\bottomrule
	\end{tabular}
\end{table*}

\subsection{Causal Quantities on Distributions}\label{sec:causal quantities}
\begin{figure} 
	\centering
	\includegraphics[scale=0.38]{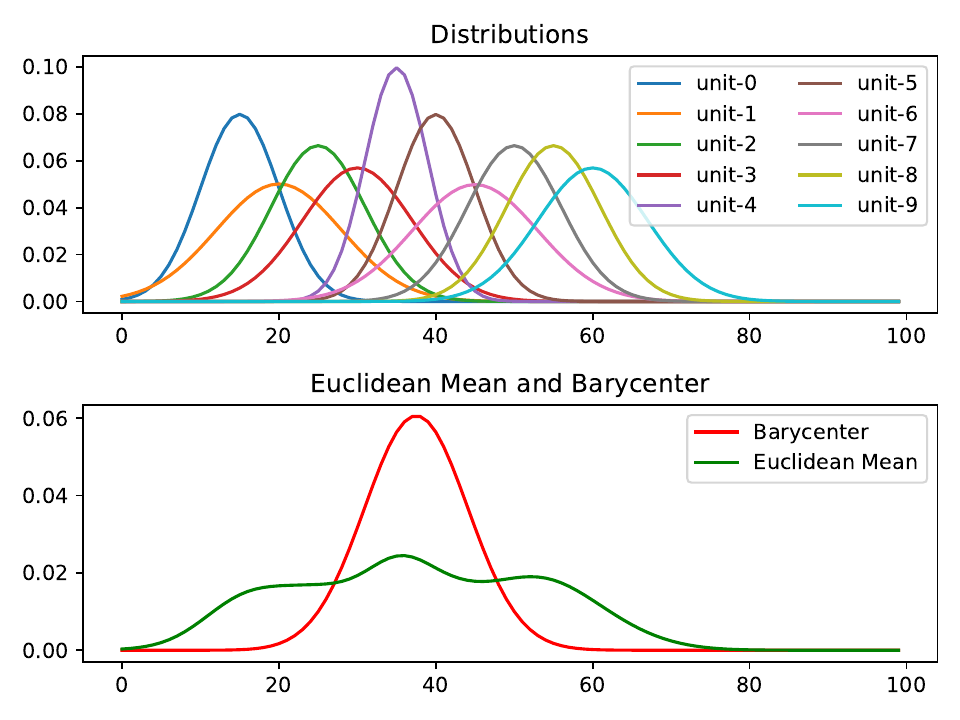}
	\caption{The Euclidean mean and Wasserstein mean (Barycenter) of 10 distributions. \label{fig:barycenter}}
\end{figure}

In our context, the realization of $\mathcal{Y}$ for each unit is a distribution. It is inappropriate to conduct causal inference in the Euclidean space as it destroys the structure of distributions. For example, Figure \ref{fig:barycenter} displays the distributions for 10 individual units (all are Gaussian distributions with different mean and variance), and the ``mean'' distribution of these 10 distributions obtained from Wasserstein metric (Barycenter) and Euclidean metric. We notice that the ``mean'' distribution cannot preserve the original Gaussian distribution structure unless the Wasserstein metric is used. We thus choose to conduct causal inference in the \textit{Wasserstein space} \citep{villani2021topics, panaretos2019statistical, feyeux2018optimal}. Here, we use the $p$-Wasserstein metric to characterize the ``distance'' between two distributions (see Definition \ref{def:Wasserstein}). In the sequel, we let the realizations of $\mathcal{Y}$, $\mathcal{Y}(d)$ reside in $\mathbb{R}$.

\begin{definition}\label{def:Wasserstein}
	Let $\mathcal{I}\subset\mathbb{R}$, $\mathcal{W}_{p}(\mathcal{I})=\{\lambda:\int_{\mathcal{I}}s^{p}d\lambda(s)<\infty\}$ ($\lambda$ is a distribution), and $\Lambda(\lambda_{1},\lambda_{2})$ be the set containing the joint distribution $\Pi(\lambda_{1}(s),\lambda_{2}(t))$ whose marginals are $\lambda_{1}$ and $\lambda_{2}$. The $p$-Wasserstein metric between two distributions $\lambda_{1}$ and $\lambda_{2}$ is
	\begin{equation*}
		{\small
			\begin{aligned}
				\mathbb{D}_{p}(\lambda_{1},\lambda_{2}) = \bigg\{\underset{\Pi\in\Lambda(\lambda_{1},\lambda_{2})}{\inf}\int_{\mathcal{I}}|s-t|^{p}d\Pi(\lambda_{1}(s),\lambda_{2}(t))\bigg\}^{\frac{1}{p}}.
			\end{aligned}
		}
	\end{equation*}
\end{definition}
$\mathbb{D}_{p}(\cdot,\cdot)$ satisfies the \textit{axioms of a metric} (i.e., non-negativity, symmetric, and triangle inequality). Usually, we set $p=2$. Next, we introduce two quantities - the \textit{causal map} and the \textit{causal effect map}.

\begin{definition}\label{def:causal map}
	The causal map of treatment $d^{i}$ is denoted as $\bigtriangleup_{d^i}$ \footnote{ $\bigtriangleup_{d^i}$ is a function and should be $\bigtriangleup_{d^i}(\cdot)$ formally. In the sequel, we use both $\bigtriangleup_{d^i}$ and $\bigtriangleup_{d^i}(\cdot)$ interchangeably.}
	such that
	\begin{equation}
		\begin{gathered}\label{eqt:causal map}
			\bigtriangleup_{d^{i}}=\mu_{d^{i}}^{-1},
		\end{gathered}
	\end{equation}
	where {\small $\mu_{d^{i}}=\underset{v\in\mathcal{W}_{2}(\mathcal{I})}{\arg\min}\;\mathbb{E}\big[\mathbb{D}_{2}(\mathcal{Y}(d^{i}),v)^{2}\big]$} is the Wasserstein barycenter/mean of units' distributions when they take the treatment $d^{i}$. The superscript ``${-1}$'' of $\mu_{d^{i}}$ is the inverse of the cumulative distribution function (CDF) or the quantile function. Hence, the causal effect map between treatment $d^{i}$ and $d^{j}$ is  
	\begin{equation}
		\begin{aligned}\label{eqt:causal effect map}
			\bigtriangleup_{d^{ij}}=\bigtriangleup_{d^{i}}-\bigtriangleup_{d^{j}}=\mu_{d^{i}}^{-1}-\mu_{d^{j}}^{-1}.
		\end{aligned}
	\end{equation}
\end{definition}
The \textit{causal effect map} in Eqn. \eqref{eqt:causal effect map} is an analogy to the ATE ($\mathbb{E}[Y(d^i)]-\mathbb{E}[Y(d^j)]$) in the literature. However,  $\bigtriangleup_{d^{i}}$, $\bigtriangleup_{d^{j}}$ and $\bigtriangleup_{d^{ij}}$ are functions, but $\mathbb{E}[Y(d^i)]$, $\mathbb{E}[Y(d^j)]$, and $\mathbb{E}[Y(d^i)]-\mathbb{E}[Y(d^j)]$ are scalars. In Table \ref{table:Comparisons}, we summarize the differences between the framework in our paper and in the literature.

\begin{remark}\ \label{remark:causal}
	\begin{enumerate}
		\item Classically, the case ``distribution over $\mathbb{R}$'' means that a realization is a point (scalar or vector) drawing from the distribution of the potential outcome, while the case ``distribution over distributions'' means that the realization is a distribution. For instance, let $\mu$ and $\sigma$ be the mean and standard deviation of a normal distribution, and $(\mu, \log\sigma) \sim \mathcal{N}(\mathbf{0}, \mathbb{I}_2)$. If the realization $(\mu, \log\sigma)$ of a unit (e.g., a consumer) is $(0.1,-0.5)$, then it means that a collection of observations (e.g., spending amounts of all orders) are drawn from $\mathcal{N}(0.1, e^{-1})$ for this unit. 
		
		\item $\bigtriangleup_{d^{i}}(\cdot)$ is a quantile function (inverse of CDF), so does $\bigtriangleup_{d^{ij}}(\cdot)$. Further, we can explore the impact of between treatment $d^{i}$ and $d^{j}$ on the distributional outcome respectively at a specific $\tau$ quantile level by $\bigtriangleup_{d^{ij}}(\cdot)$, i.e.,
		\begin{equation}
			{\small
				\begin{aligned}
					\bigtriangleup_{d^{ij}}(\tau)=\bigtriangleup_{d^{i}}(\tau)-\bigtriangleup_{d^{j}}(\tau)=\mu_{d^{i}}^{-1}(\tau)-\mu_{d^{j}}^{-1}(\tau).
				\end{aligned}
			}
		\end{equation}
		Note that $\bigtriangleup_{d^{ij}}(\tau)$ differs from the quantile treatment effect (QTE) in the literature (e.g., \cite{machado2005counterfactual, chernozhukov2005iv}). $\bigtriangleup_{d^{ij}}(\tau)$ is the $\tau$-quantiles difference of the \textbf{barycenters} under treatments $d^{i}$ and $d^{j}$, but QTE is the $\tau$-quantiles difference of the \textbf{potential outcome distributions} under two treatments. It is thus inappropriate to compare them or study $\bigtriangleup_{d^{ij}}(\tau)$ using the approaches in the QTE literature. The visualized difference of the two quantities is given in Figure \ref{fig:ATE and QTE} and \ref{fig:causal effect map}.
		
	\end{enumerate}
\end{remark}

We need to ensure $\bigtriangleup_{d^i}$ is \textit{identifiable} such that we can estimate it from an observed dataset. It is also necessary to simplify the calculation of $\mu_{d^i}$ to address the computational complexity of optimal transport. Proposition \ref{prop:causal effect expectation form} states an equivalent form of $\bigtriangleup_{d^i}$ without computing optimization and guarantees that we can estimate it from the observed dataset:

\begin{proposition}\label{prop:causal effect expectation form}
	Given the conditions in Definition \ref{def:Wasserstein} and \ref{def:causal map}, and Assumptions (1) - (3) hold, we have (1) $\bigtriangleup_{d^i}=\mathbb{E}\big[\mathcal{Y}(d^{i})^{-1}\big]$; (2) $\bigtriangleup_{d^i}$ is \textit{identifiable}.
\end{proposition}
The first assertion gives a simpler way to compute $\bigtriangleup_{d^i}$, while the second assertion ensures that $\bigtriangleup_{d^i}$ is identifiable. We defer the proofs to Appendix D.

\subsection{Estimators}\label{sec:estimators}
Similar to the causal inference methods in the literature \citep{horvitz1952generalization, chernozhukov2018double}, we also propose three estimators to compute the causal map $\bigtriangleup_{d^i}$, namely (1) \textit{Direct Regression (DR) estimator ($\bigtriangleup_{d^i;DR}$)}, (2) \textit{Inverse Probability Weighting (IPW) estimator ($\bigtriangleup_{d^i;IPW}$)}, and (3) \textit{Double Machine Learning (DML) estimator ($\bigtriangleup_{d^i;DML}$)}. Let $\pi_{d^i}(\mathbf{X})=\mathbb{P}\{D=d^{i}|\mathbf{X}\}$ and $m_{d^i}(\mathbf{X})=\mathbb{E}\big[\mathcal{Y}^{-1}|D=d^{i},\mathbf{X}\big]$. Given that there are $N$ units. The estimators are:
{\small
	\begin{gather}
		\bigtriangleup_{d^i;DR}=\frac{1}{n}\underset{s=1}{\overset{n}{\sum}}m_{d^i}(\mathbf{X}_s) \label{eqt:DR estimator}\\
		\bigtriangleup_{d^i;IPW}=\frac{1}{n}\underset{s=1}{\overset{n}{\sum}}\frac{\mathbf{1}_{\{D_s=d^{i}\}}}{\pi_{d^i}(\mathbf{X}_s)}(\mathcal{Y}_{s}^{-1}) \label{eqt:IPW estimator}\\
		\bigtriangleup_{d^i;DML}=\frac{1}{n}\underset{s=1}{\overset{n}{\sum}}\big[m_{d^i}(\mathbf{X}_{s})+\frac{\mathbf{1}_{\{D_s=d^{i}\}}}{\pi_{d^i}(\mathbf{X}_s)}(\mathcal{Y}_{s}^{-1}-m_{d^i}(\mathbf{X}_s))\big]. \label{eqt:DML estimator}
	\end{gather}
	
}

\subsection{Theory and Algorithm}\label{sec:algorithm}
In practical scenarios, when using all the available units to train the regression function $m_{d^i}(\mathbf{X}_s)$ and propensity score function $\pi_{d^i}(\mathbf{X}_s)$, there is a risk of over-fitting. To mitigate this issue, a cross-fitting technique, as introduced by \cite{chernozhukov2018double}, is commonly employed. Along this way, we also need to obtain the cross-fitting estimators of $\bigtriangleup_{d^i}$ according to Eqns. \eqref{eqt:DR estimator}, \eqref{eqt:IPW estimator}, and  \eqref{eqt:DML estimator}.

We split the $N$ units into $K$ disjoint groups. Let the $k^{\text{th}}$ group be $\mathcal{D}_{k}$ of size $N_k$, $\forall k=1,\cdots, K$. Denoting $\mathcal{D}_{-k}={\cup}^K_{r=1, r\neq k}\mathcal{D}_{r}$, we use $\mathcal{D}_{-k}$ to obtain  $\hat{m}_{d^i}^k(\mathbf{X})$, $\hat{\pi}^{k}_{d^i}(\mathbf{X})$, which are the estimations of $m_{d^i}^k(\mathbf{X})$, $\pi^{k}_{d^i}(\mathbf{X})$ for the $k^{\text{th}}$ group. $\mathcal{\hat{Y}}$ is the empirical estimation of $\mathcal{Y}$. We then use $\mathcal{D}_{k}$ to compute the estimation of $\bigtriangleup^k_{d^i}$ (i.e., $\hat{\bigtriangleup}_{d^i;DR}^{k}$, $\hat{\bigtriangleup}_{d^i;IPW}^{k}$, and $\hat{\bigtriangleup}_{d^i;DML}^{k}$) according to Eqns. \eqref{eqt:cross-fit estimator DR}, \eqref{eqt:cross-fit estimator IPW}, and \eqref{eqt:cross-fit estimator DML} respectively. We thus define $\hat{\bigtriangleup}_{d^i;DR}^{k}$, $\hat{\bigtriangleup}_{d^i;IPW}^{k}$, and $\hat{\bigtriangleup}_{d^i;DML}^{k}$ such that

{\small
	\begin{gather}
		\hat{\bigtriangleup}_{d^i;DR}^{k}=\frac{1}{N_{k}}\underset{s\in\mathcal{D}_{k}}{\overset{}{\sum}}\hat{m}_{d^i}^{k}(\mathbf{X}_{s}) \label{eqt:cross-fit estimator DR}\\
		\hat{\bigtriangleup}_{d^i;IPW}^{k}=\frac{1}{N_{k}}\underset{s\in\mathcal{D}_{k}}{\overset{}{\sum}}\frac{\mathbf{1}_{\{D_{s}=d^{i}\}}}{\hat{\pi}^k_{d^i}(\mathbf{X}_{s})}\mathcal{\hat{Y}}_{s}^{-1} \label{eqt:cross-fit estimator IPW}\\
		\hat{\bigtriangleup}_{d^i;DML}^{k}=\frac{1}{N_{k}}\underset{s\in\mathcal{D}_{k}}{\overset{}{\sum}}\big[\hat{m}_{d^i}^{k}(\mathbf{X}_{s})+\frac{\mathbf{1}_{\{D_{s}=d^{i}\}}}{\hat{\pi}^k_{d^i}(\mathbf{X}_{s})}(\mathcal{\hat{Y}}_{s}^{-1}-\hat{m}_{d^i}^{k}(\mathbf{X}_{s}))\big].\label{eqt:cross-fit estimator DML}
	\end{gather}
}

Denoting $w\in\{DR,IPW,DML\}$, the cross-fitting estimators are $\hat{\bigtriangleup}_{d^i;w}$ such that
\begin{equation}
	\begin{aligned}\label{eqt:cross-fit estimator final}
		\hat{\bigtriangleup}_{d^i;w}=\underset{k=1}{\overset{K}{\sum}}\frac{N_{k}}{N}\hat{\bigtriangleup}^k_{d^i;w}. 
	\end{aligned}
\end{equation}  

\noindent We study the consistency of $\hat{\bigtriangleup}_{d^i;w}$. When $w=DR$ or $IPW$, the results are deferred to Appendix C. When $w=DML$, the consistency result is given in Theorem \ref{thm:root N consistent} while the proofs and the notational meanings are deferred to Appendix D.
\begin{theorem}\label{thm:root N consistent}
	Let $\tilde{m}_{d^i}^{k}(\mathbf{X})$ ($\hat{m}_{d^i}^{k}(\mathbf{X})$) be the estimate of $\mathbb{E}[\mathcal{Y}^{-1}|D=d^{i},\mathbf{X}]$ using the true $\mathcal{Y}$ (estimated $\hat{\mathcal{Y}}$) based on $\mathcal{D}_{-k}$. Suppose that, for any $k$, $\rho_{\pi}^{4}=\mathbb{E}[|\hat{\pi}_{d^i}^{k}(\mathbf{X})-\pi_{d^i}(\mathbf{X})|^{4}]$, $\rho_{m}^{4}=\max\{\vvvert \tilde{m}_{d^i}^{k}-m_{d^i}\vvvert^{4},\;1\leq i\leq r\}=\max\{[\int \|\tilde{m}_{d^{i}}^{k}(\mathbf{x})-m_{d^{i}}(\mathbf{x})\|^{2}dF_{\mathbf{X}}(\mathbf{x})]^{2},\;1\leq i\leq r\}$. Under the convergence assumptions in Appendix D, we have
	\begin{enumerate}
		\item $\|\hat{\bigtriangleup}_{d^i;DML}-\bigtriangleup_{d^i}\|=O_{P}(N^{-\frac{1}{2}}+N^{-\frac{1}{2}}\rho_{\pi}+N^{-\frac{1}{2}}\rho_{m}+\rho_{\pi}\rho_{m})$\label{thm:root N consistent1}.
		\item If $\rho_{m}\rho_{\pi}=o(N^{-\frac{1}{2}})$, $\rho_{m}=o(1)$ and $\rho_{\pi}=o(1)$, then $\sqrt{N}\big(\hat{\bigtriangleup}_{d^i;DML}-\bigtriangleup_{d^i}\big)$ converges weakly to a centred Gaussian process.\label{thm:root N consistent2}
	\end{enumerate}
\end{theorem}
Theorem \ref{thm:root N consistent} not only gives the consistency of $\hat{\bigtriangleup}_{d^i;DML}$, but also gives the convergence speed of $\hat{\bigtriangleup}_{d^i;DML}$. It is indeed a $\sqrt{N}$-consistent estimator.

We can also investigate the $\sqrt{N}$-consistency of the DR or IPW estimators. In fact, we can obtain the desired results by setting $\mathbf{1}_{\{D=d^{i}\}}=0$ and $(m_{d^{i}},\hat{m}_{d^{i}}^{k},\tilde{m}_{d^{i}}^{k})=(0,0,0)$ in the proofs of Theorem \ref{thm:root N consistent} respectively. Last but not least, we summarize the steps of computing $\hat{\bigtriangleup}_{d^i;w}$ in Algorithm \ref{alg:estimators}.
{\small
	\begin{algorithm} 
		\caption{Computations of $\hat{\bigtriangleup}_{d^i;w}$}\label{alg:estimators}
		\begin{algorithmic}[1]
			\REQUIRE The observations of $(D_{s},\mathbf{X}_{s},\mathcal{Y}_{s})_{s=1}^{N}$.
			\ENSURE $\hat{\bigtriangleup}_{d^i;w}$ for $w\in\{DR,IPW,DML\}$.
			\STATE Split $(D_{s},\mathbf{X}_{s},\mathcal{Y}_{s})_{s=1}^{N}$ to $K$ disjoint units groups $\mathcal{D}_{k}$ of size $N_{k}$ and form $\mathcal{D}_{-k}$.
			\STATE Estimate $\hat{\mathcal{Y}}_s^{-1}$ for each unit $s$.
			\FOR {$k \gets 1$ to $K$}
			\STATE \parbox[t]{\dimexpr\linewidth-\algorithmicindent}{Regress $D$ w.r.t. $\mathbf{X}$ based on $\mathcal{D}_{-k}$ and obtain $\hat{\pi}_{d^i}^k$.}\label{alg:probability estimation}
			\STATE \parbox[t]{\dimexpr\linewidth-\algorithmicindent}{Regress $\hat{\mathcal{Y}}^{-1}$ w.r.t. $(D,\mathbf{X})$ based on $\mathcal{D}_{-k}$  and obtain $\hat{m}_{d^i}^k$.}\label{alg:function estimation}
			\STATE \parbox[t]{\dimexpr\linewidth-\algorithmicindent}{Compute $\hat{\bigtriangleup}^k_{d^i;w}$ using Eqns. \eqref{eqt:cross-fit estimator DR}, \eqref{eqt:cross-fit estimator IPW} and \eqref{eqt:cross-fit estimator DML} according to $w$.}
			\ENDFOR
			\STATE{Compute $\hat{\bigtriangleup}_{d^i;w}$} using Eqn. \eqref{eqt:cross-fit estimator final}.
		\end{algorithmic}
	\end{algorithm}
}

\subsection{Models}
To estimate the target quantity $\bigtriangleup_{d^i}$, we need to estimate several nuisance parameters accurately, e.g., $\mathcal{Y}^{-1}$, $\pi_{d^i}(\mathbf{X})$, and $m_{d^i}(\mathbf{X})$. First, to estimate $\mathcal{Y}^{-1}$, we can estimate $\mathcal{Y}$ empirically and invert the estimated $\mathcal{Y}$ (CDF) for each unit to get the $\mathcal{\hat{Y}}^{-1}$. Second, $\pi_{d^i}(\mathbf{X})$ is the \textit{propensity score} that can be estimated using the multi-class logistic regression, random forest classifier, or feed-forward networks. Finally, we can estimate the regression function $m_{d^i}(\mathbf{X})$ by regressing the outcome $\mathcal{Y}^{-1}$ on treatment $D$ and covariates $\mathbf{X}$ via a \textit{functional-on-scalar regression}. The first two quantities can be well estimated using the classical approaches. On the other hand, the third quantity, $m_{d^i}(\mathbf{X})$, is difficult to estimate accurately using the classical functional regression approach. Specifically, the classical functional regression \citep{ramsay2005fitting} assumes that the regression equation between outcome $\mathcal{Y}^{-1}$ and predictors $(D, \mathbf{X})$ can be approximated by a finite series of some pre-determined basis functions, i.e.,
\begin{equation} \label{eqt:functional regression}
	\mathcal{Y}^{-1}(t)=D\sum_{l=1}^v \gamma_{0 l} \phi_l(t)+\sum_{j=1}^n X^{j}\left(\sum_{l=1}^v \gamma_{j l} \phi_l(t)\right)+\epsilon(t),
\end{equation}
where $\mathcal{Y}^{-1}(t)$ is the response function; $(D, \mathbf{X})=[D, X^1, \cdots,X^j, \cdots, X^n]$ are predictors; $\{\phi_1, \dots, \phi_v\}$ are basis functions, e.g., B-spline basis;  $\gamma_{jl}$ with $0 \leq j \leq n$ and $1 \leq l \leq v$ are regression parameters; and $\epsilon(t)$ is the noise term.

\begin{figure}
	\centering
	\includegraphics[width=0.5\columnwidth]{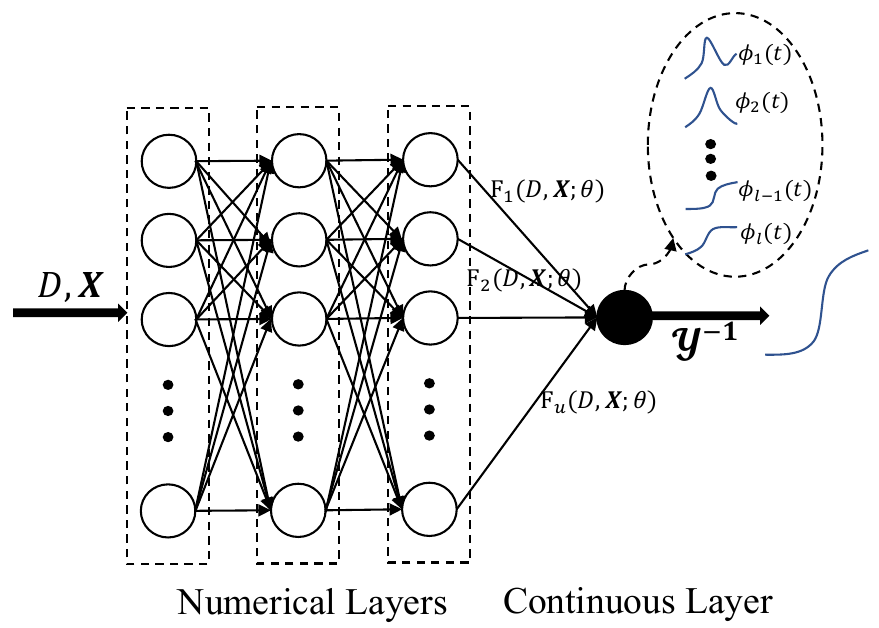}
	\caption{The proposed NFR Net. \label{fig:models}}
\end{figure}

However, the relation between $\mathcal{Y}^{-1}(t)$ and $(D, \mathbf{X})$ may not be additive as in Eqn. \eqref{eqt:functional regression}. Generally, the relationship is non-linear and complex. To this end, we design \textbf{N}eural \textbf{F}unctional \textbf{R}egression (\textbf{NFR}) Net to address this issue. The NFR Net consists of two parts: (1) the \textit{numerical layers}, and (2) the \textit{continuous layer} (see Figure \ref{fig:models}). Under our framework and settings, the numerical layers aim to learn the $u$ representations $\mathbf{\mathsf{F}}(D, \mathbf{X}; \theta)=[\mathsf{F}_1(D, \mathbf{X}; \theta), \cdots, \mathsf{F}_u(D, \mathbf{X}; \theta)]^\top$, where each $\mathsf{F}_i(D, \mathbf{X}; \theta), 1 \leq i \leq u$ is a linear coefficient to constitute the target distribution. The representations $\mathbf{\mathsf{F}}(D, \mathbf{X};\theta)$ is then processed by a continuous layer to output a function $\tilde{\mathcal{Y}}^{-1}$, i.e.,
\begin{equation}
	\tilde{\mathcal{Y}}^{-1}(t; \theta, \{\gamma_{ij}\})=\sum_{i=1}^u \mathsf{F}_i \left(D, \mathbf{X}; \theta \right) \sum_{j=1}^v \gamma_{i j} \phi_j(t), 
\end{equation}
where $\{\gamma_{ij}\}$ now are trainable parameters and $\{\phi_j(t)\}$ are pre-defined basis functions. 

The model can be trained as follows: let $L$ be the loss metric (e.g., $L_1$ or $L_2$ loss), and our task is equivalent to finding the optimal $\theta,\{\gamma_{ij}\}$ by minimizing the loss function $\mathcal{L}(\theta,\{\gamma_{ij}\})$:
\begin{equation}
	{\small
		\begin{aligned}\label{eqt:optimization}
			\min_{\theta, \{\gamma_{ij}\}}\mathcal{L}(\theta, \{\gamma_{ij}\}):=\int L(\tilde{\mathcal{Y}}^{-1}(t; \theta, \{\gamma_{ij}\}),\hat{\mathcal{Y}}^{-1}(t))dt.
		\end{aligned}
	}
\end{equation}
In practice, we can estimate the integral using the trapezoidal rule/Simpson's rule by taking any number of discrete quantile points $t$.

\begin{figure}
	\centering
	\begin{minipage}{.43\columnwidth}
		\centering
		\includegraphics[width=0.8\columnwidth]{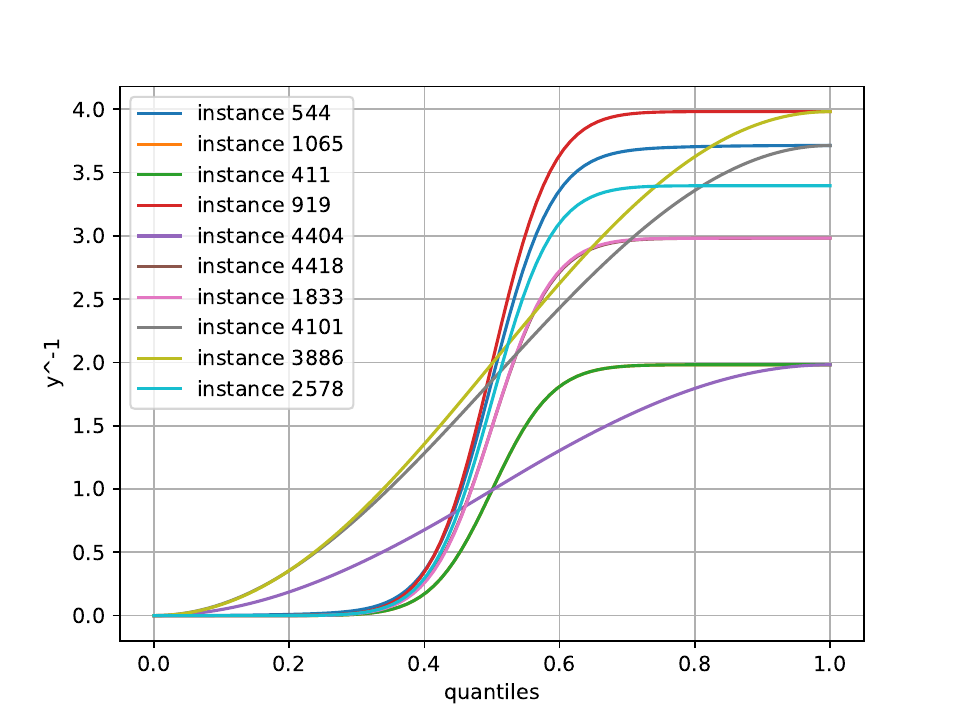}
		\caption{10 instances of simulated quantile function ($\mathcal{Y}^{-1}$). \label{fig:instance}}
	\end{minipage}
	\quad
	\begin{minipage}{.48\columnwidth}
		\centering
		\includegraphics[width=0.8\columnwidth]{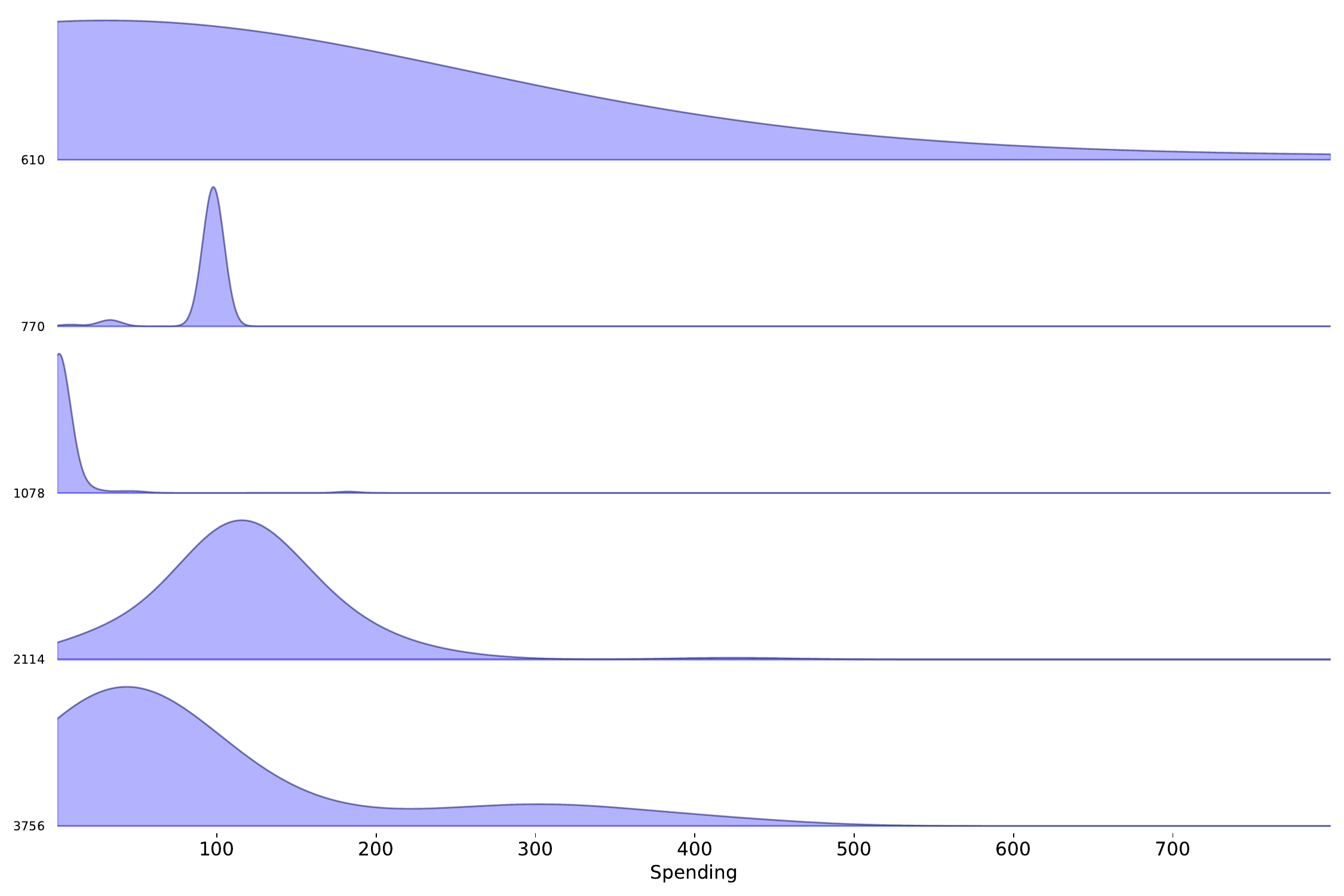}
		\caption{The spending distribution of 5 consumers from a e-commerce platform. \label{fig:spending distribution}}
	\end{minipage}
\end{figure}

\section{Synthetic Experiment} \label{sec:synthetic experiment}
\paragraph{Data Generation Process} Since the ground truth is unavailable in the real dataset, we simulate data using the following data generation process for the $s^{\mathrm{th}}$ unit to test our proposed model similar to many other causal inference studies:
\begin{subequations}
	{\small
		\begin{align}
			&\mathcal{Y}^{-1}_s(D_{s})=c+(1-c)(\mathbb{E}[D]+\sqrt{D_s})\times \nonumber\\
			&\;\;\;\;\;\;\;\;\;\;\;\;\;\;\;\;\;\;\underset{j=1}{\overset{\frac{n}{2}}{\sum}} \frac{\exp(X^{2j-1}_{s}X^{2j}_s)}{\underset{k=1}{\overset{\frac{n}{2}}{\sum}}\exp(X^{2k-1}_{s}X^{2k}_{s})}\mathbf{B}^{-1}(\alpha_j,\beta_j)+\epsilon_s,\label{eqt:DGP1}\\
			&\mathbb{P}\{D_s=d\;|\;\mathbf{X}_s\} = \frac{\exp(\gamma_{d}^{\top} \mathbf{X}_{s})}{\underset{w=1}{\overset{r}{\sum}}\exp(\gamma_{w}^{\top} \mathbf{X}_{s})}.\label{eqt:DGP2}
		\end{align}
	}\noindent
\end{subequations}
\begin{table}[t]
	\caption{MAE between true and estimated causal effect maps under various methods of regressing $\mathcal{\hat{Y}}^{-1}$ w.r.t. $(D,\mathbf{X})$ (mean $\pm$ standard deviation with 100 trials). Best results are in bold.}\label{results}
	\centering
	\resizebox{0.65\columnwidth}{!}{
		\begin{tabular}{ccccccc}
			\toprule
			\textbf{} & \textbf{Q=10\%} & \textbf{Q=30\%} & \textbf{Q=50\%} & \textbf{Q=70\%} & \textbf{Q=90\%} & \textbf{Average} \\\midrule
			\textbf{} & \multicolumn{6}{c}{\textbf{DR}} \\ \midrule
			\textbf{Lasso} 			& 0.066 & 0.064 & 0.121 & 0.183 & 0.184 & 0.124 $\pm$ 0.053 \\
			\textbf{Ridge} 			& 0.066 & 0.075 & 0.093 & 0.167 & 0.189 & 0.118 $\pm$ 0.050 \\
			\textbf{Elastic net} 	& 0.066 & 0.064 & 0.121 & 0.183 & 0.184 & 0.124 $\pm$ 0.053 \\
			\textbf{MCP} 			& 0.066 & 0.075 & 0.093 & 0.167 & 0.189 & 0.118 $\pm$ 0.050 \\
			\textbf{D/DML} 			& 0.010 & 0.058 & 0.299 & 0.517 & 0.615 & 0.300 $\pm$ 0.240 \\
			\textbf{NFR} 			& 0.025 & 0.104 & 0.039 & 0.053 & 0.040 & \textbf{0.052} $\pm$ \textbf{0.027} \\ \midrule
			\textbf{} & \multicolumn{6}{c}{\textbf{IPW}} \\ \midrule
			\textbf{Lasso} 			& 0.010 & 0.061 & 0.070 & 0.046 & 0.035 & 0.044 $\pm$ 0.021 \\
			\textbf{Ridge} 			& 0.010 & 0.061 & 0.070 & 0.045 & 0.034 & 0.044 $\pm$ 0.021 \\
			\textbf{Elastic net} 	& 0.010 & 0.061 & 0.070 & 0.046 & 0.037 & 0.045 $\pm$ 0.021 \\
			\textbf{MCP} 			& 0.010 & 0.061 & 0.069 & 0.047 & 0.035 & 0.044 $\pm$ 0.021 \\
			\textbf{D/DML} 			& 0.010 & 0.061 & 0.070 & 0.046 & 0.035 & 0.044 $\pm$ 0.021 \\
			\textbf{NFR} 			& 0.010 & 0.061 & 0.071 & 0.042 & 0.033 & 0.044 $\pm$ 0.021 \\ \midrule
			\textbf{} & \multicolumn{6}{c}{\textbf{DML}} \\ \midrule
			\textbf{Lasso} 			& 0.012 & 0.064 & 0.086 & 0.020 & 0.006 & 0.038 $\pm$ 0.032 \\
			\textbf{Ridge} 			& 0.010 & 0.063 & 0.068 & 0.024 & 0.005 & 0.034 $\pm$ 0.027 \\
			\textbf{Elastic net} 	& 0.012 & 0.064 & 0.085 & 0.020 & 0.006 & 0.038 $\pm$ 0.032 \\
			\textbf{MCP} 			& 0.010 & 0.063 & 0.068 & 0.024 & 0.005 & 0.034 $\pm$ 0.027 \\
			\textbf{D/DML} 			& 0.010 & 0.063 & 0.071 & 0.025 & 0.013 & 0.037 $\pm$ 0.026 \\
			\textbf{NFR} 			& 0.010 & 0.062 & 0.067 & 0.025 & 0.003 & \textbf{0.034} $\pm$ \textbf{0.026} \\ \bottomrule
		\end{tabular}
	}
	\begin{tablenotes} 
		\item \scriptsize The IPW results are similar because the same classifier (random forest) is used to get the propensity scores.
	\end{tablenotes} 
\end{table}

In our experiment, we set $n=10$. We assume that covariates $X^1, X^2 \sim \mathcal{N}(-2,1)$, $X^3, X^4 \sim \mathcal{N}(-1,1)$, $X^5, X^6 \sim \mathcal{N}(0,1)$, $X^7, X^8 \sim \mathcal{N}(1,1)$, $X^9, X^{10} \sim \mathcal{N}(2,1)$, and $\epsilon_s \sim \mathcal{N}(0,0.05)$. $\mathbf{B}^{-1}(\alpha,\beta)$ is the inverse CDF of Beta distribution with the shape parameters $\alpha$ and $\beta$. We select $5$ inverse Beta CDFs, where each one has different parameters to ensure the complexity of the distribution function. The treatment $D$ takes the value in $\{d^1, d^2, d^3, d^4, d^5\}$ with a softmax distribution. $c \in [0, 1]$ is the constant that controls the strength of the causal relationship between treatment $D$ and outcome distribution $\mathcal{Y}^{-1}$. In one experiment, $5,000$ instances are generated according to Eqns. \eqref{eqt:DGP1} - \eqref{eqt:DGP2}. For each unit $s$, $100$ observations are sampled from the inverse CDF using the inverse transform sampling method. Figure \ref{fig:instance} summarizes $10$ simulated instances, indicating that the inverse CDF of each instance varies widely. 

\paragraph{Baselines} In our experiment, we consider two aspects of potential baseline methods. The first aspect is from the statistical field, where approaches such as those presented in \cite{lin2023causal} assume a linear relationship between the functional output and the scalar input. They utilize regularization techniques like lasso, ridge, and elastic net to estimate the causal effect map. Additionally, \cite{chen2016variable} addresses situations with a large number of covariates by using the group minimax concave penalty (MCP) for variable selection and fitting. However, these methods inherently assume a linear form between the functional output and scalar input, possibly overlooking the presence of nonlinear relationships in the data. The second aspect is from the deep learning field, where we compare our model with classical Double/debiased machine learning (D/DML) proposed in \cite{chernozhukov2018double}. This approach introduces a DML estimator to investigate the causal impact of scalar input on scalar outcome. To model the functional outcome, we conduct independent regressions at interesting quantiles using a standard MLP. Subsequently, we concatenate all the quantile counterfactuals to form a distribution.

\paragraph{Experiment Setting} The classification and functional regression models are trained separately. $5,000$ generated instances are trained using $5$-fold cross-fitting, i.e., $4,000$ instances are used to train, and $1,000$ instances are used to obtain the three estimators (i.e., DR, IPW, and DML estimator). At last, we average the obtained estimators from the $5$ folds as the final results. In the classification task, we use the same classifier (i.e., random forest) to compute IPW for all the estimators. 
The training details are given in Appendix E.

\paragraph{Evaluation Metric} Since $\mathcal{L}(\theta,\{\gamma_{ij}\})$ in \eqref{eqt:optimization} is continuous, we discretize it and compare the mean absolute error (MAE) between true causal effect map $\bigtriangleup_{d^{ij}} (1 \leq i, j \leq 5)$ (computed from Eqns. \eqref{eqt:DGP1} - \eqref{eqt:DGP2}) and estimated causal effect map $\hat{\bigtriangleup}_{d^{ij}}$ on 5 quantiles with levels ranging from 10\% to 90\%. We repeat the experiment 50 times to report the mean and standard deviation of MAE.

\paragraph{Experiment Results}  
Table \ref{results} presents a summary of the experiment results. We observe several key findings: Firstly, NFR Net demonstrates superior performance compared to all statistical models, particularly on the DR methods. This result can be attributed to the capability of our proposed model to capture non-linear patterns between covariates and the outcome distribution effectively. Secondly, NFR Net outperforms the D/DML method. The advantage stems from our ability to model the outcome as a function. In contrast, D/DML treats each quantile as independent scalar points, overlooking the continuous structure of the distribution. Lastly, DML can utilize the IPW estimator to correct most of the bias in the DR estimator, and the DML estimator demonstrates improved robustness compared to both the DR and IPW estimators. 

\section{Empirical Experiment}
E-commerce platforms face a significant challenge in comprehending the impact of credit lines on consumer spending patterns, particularly in terms of the shift in spending distribution caused by changes in the credit lines. To address this issue, we employ our approach by leveraging data collected from a large e-commerce platform. The platform assigns distinct credit lines to users based on various factors such as income, age, and past behaviors like shopping and default behaviors. Besides, the platform provides users with an interest-free, one-month loan option for their purchases, with the condition that the total loan amount must not exceed their assigned credit lines.

We collect data from 4,043 platform users. The data comprises various variables, such as demographic information (e.g., age, income, and location), purchasing behaviors (e.g., the total number of orders, the amount paid for each order), and financial information (e.g., credit lines assigned by the platform, the total number of loans, and the presence of default records). Appendix F displays a detailed statistical description. All the paid amounts of orders constitute a unique spending distribution for each user (e.g., Figure \ref{fig:spending distribution}). In our empirical study, we investigate the causal maps when the credit lines take values as Low (from 0 to 9,000), Middle (from 9,000 to 15,000), and High (higher than 15,000). 

\begin{table*}[tbh]
	\caption{The results of the causal map of three treatments at 9 quantiles (mean and 95\% CI).}
	\centering
	\resizebox{0.9\textwidth}{!}{
		\begin{tabular}{cccccc} 
			\toprule
			\textbf{Quantiles} & \textbf{Low} & \textbf{Middle} & \textbf{High}& \textbf{Low$\rightarrow$Middle}& \textbf{Low$\rightarrow$High}  \\ \midrule
			\textbf{10\%}      & 28.0 (27.8, 28.3)      & 29.9 (29.8, 30.0)     & 30.6 (30, 31.2)        &    6.79\%$\uparrow$	&9.29\%$\uparrow$    \\
			\textbf{20\%}      & 43.6 (43.4, 43.9)      & 47.4 (47.1, 47.8)     & 48.8 (47.9, 49.5)      &    8.72\%$\uparrow$	&11.93\%$\uparrow$             \\
			\textbf{30\%}      & 58.7 (58.3, 59.0)      & 65.5 (65.1, 65.9)     & 67.5 (66.5, 68.4)      &    11.58\%$\uparrow$	&14.99\%$\uparrow$   \\
			\textbf{40\%}      & 75.2 (74.7, 75.6)      & 86.8 (86.3, 87.2)     & 91.0 (90.0, 92.1)      &    15.43\%$\uparrow$   &21.01\%$\uparrow$ \\
			\textbf{50\%}      & 94.9 (94.3, 95.6)      & 115.8 (114.9, 116.9)  & 122.4 (121.1, 123.8)   &    22.02\%$\uparrow$	&28.98\%$\uparrow$   \\
			\textbf{60\%}      & 119.0 (118.2, 119.7)   & 150.8 (149.7, 152.0)  & 170.8 (167.4, 174.7)   &    26.72\%$\uparrow$	&43.53\%$\uparrow$   \\
			\textbf{70\%}      & 155.1 (153.7, 156.4)   & 207.0 (205.6, 208.5)  & 256.0 (251.8, 261.5)   &    33.46\%$\uparrow$	&65.05\%$\uparrow$   \\
			\textbf{80\%}      & 212.9 (210.8, 214.6)   & 325.6 (323.2, 328.3)  & 433.0 (424.4, 442.7)   &    52.94\%$\uparrow$	&103.38\%$\uparrow$  \\
			\textbf{90\%}      & 381.0 (374.1, 386.7)   & 654.5 (650.7, 658.4)  & 1020.3 (1003.8, 1036.9)&    71.78\%$\uparrow$	&167.80\%$\uparrow$  \\ \bottomrule
		\end{tabular}
	}
	\label{table:quantiles causal map}
\end{table*}

In Table \ref{table:quantiles causal map}, we give 9 percentiles of the causal map $\bigtriangleup_{High}$, $\bigtriangleup_{Middle}$, and $\bigtriangleup_{Low}$ of all the consumers' spending distributions if they are assigned to High, Middle, and Low credit lines, respectively. 
Generally, the lower quantile of spending distribution stands for life necessities, while the higher quantiles represent luxury goods. Our findings support prior research \citep{aydin2022consumption, soman2002effect}, revealing a positive correlation between credit lines and spending since the causal effect maps $\bigtriangleup_{High}-\bigtriangleup_{Low}$ and $\bigtriangleup_{Middle}-\bigtriangleup_{Low}$ are always positive at all quantiles. Additionally, we observe that such an effect is heterogeneous across different quantiles. Specifically, when the credit lines increase, the spending on higher quantiles (e.g., higher than 70\%) grows significantly while the spending on lower quantiles increases relatively slowly. For example, when credit lines change from Low to High, the spending at 90\% quantile increases from 381.0 to 1020.3 (increasing about 167.8\%) while the spending at 10\% quantile only increases from 28.0 to 30.6 (increasing about 9.3\%).  This suggests that users tend to increase their spending on luxury goods or services when they are able to access higher credit.


\section{Conclusion}\label{sec:conclusion}
\vspace{-0.1cm}
We study the causal inference on distributional outcomes with multiple treatments in the Wasserstein space. Our target quantity, the causal effect map, is the analogy to ATE in classical causal inference literature. We then propose three estimators, i.e., DR, IPW, and DML estimators, and study their asymptotic properties. Our proposed NFR Net captures complex patterns among covariates, treatments, and functional outcomes, which is verified by the synthetic experiment. Moreover, we apply it to a credit dataset and explore the causal relationship between credit lines and spending distributions. We find that when credit lines increase, the spending at every quantile level increases, with a more significant change at higher quantiles.

Generally, the credit lines is measured continuously, and a potential future research direction involves investigating causal inference in the context of continuous treatment. Additionally, the realized distribution can be multivariate, such as the joint distribution of spending behavior and credit risk, providing an opportunity to explore such scenarios. 


\section{Acknowledgments}
Qi Wu acknowledges the support from The CityU-JD Digits Joint Laboratory in Financial Technology and Engineering; The Hong Kong Research Grants Council [General Research Fund 11219420/9043008 and 11200219/9042900]; The HK Institute of Data Science. The work described in this paper was partially supported by the InnoHK initiative, The Government of the HKSAR, and the Laboratory for AI-Powered Financial Technologies.

\bibliographystyle{apalike}
\bibliography{template}

\clearpage
\appendix
\section{Causal Assumptions}\label{sec:causal assumptions}

We introduce the assumptions that are used for causal inference in Section \ref{sec:causal framework}. Here, we explain the importance of each assumption.
\begin{description}
	\item[Consistency] It assures that the observed outcome is due to the assigned intervention that allows us to examine the target quantities from the observable dataset.
	\item[Ignorability/Unconfoundness] It has two meanings. First, if two units have the same $\mathbf{X}$, then the joint distributions of $(\mathcal{Y}(d^{1}),\cdots,\mathcal{Y}(d^{r}))$ conditioning on the covariates $\mathbf{X}$ and the treatment assignment of the two units are the same. Second, the treatment assignment mechanism should be the same if two units have the same $\mathbf{X}$. This assumption is not stronger than the
	assumption for scalar outcomes since the outcome distribution should be understood as a whole part. In our context, this assumption indicates that given one’s covariates, the spending behavior under different treatments is independent of the treatment assignment.

	\item[Positivity] It assures that every treatment has the chance to be assigned to the units. Indeed, if the probability of a particular treatment is $0$, then it is impossible to evaluate the effects due to the treatment.
\end{description}
Apart from the Consistency, Ignorability/Unconfoundness, and Overlap assumptions, we also assume the SUTVA assumption:
\begin{assumption}
	It contains two parts:
	\begin{enumerate}
		\item The potential outcome of a unit is not influenced by the treatment assignment to other units.
		\item For each unit, there are no different forms of treatment levels that lead to different potential outcomes.
	\end{enumerate}
\end{assumption}
Indeed, Statement (1) assures that the potential outcome of a unit is due to the treatment level that the unit receives but not the treatment assignment to other units. Statement (2) ensures that each treatment level should be clearly characterized. Concretely, consider the case that we are interested in the effects of taking Aspirin. If the treatment variable is binary (either taking Aspirin or not), then every patient who takes Aspirin should take the same type of Aspirin and dosage. No patients are allowed to take different dosages. If the dosage is essential, then the treatment should be the dosages of Aspirin a patient takes but not dichotomous. 

\section{Causal Quantities on Distributions}\label{sec:causal quantity difference}

We emphasize that ($\mathcal{Y}, \mathcal{Y}(d)$) and ($Y, Y(d)$) have different meanings. ($Y, Y(d)$) is the outcome variable, and its realization is scalar values, while ($\mathcal{Y}, \mathcal{Y}(d)$) is the outcome variable such that its realization is distribution. In the main paper, we give a parametric example illustrating the meaning of ($\mathcal{Y}, \mathcal{Y}(d)$). Here, we give two concrete examples illustrating the differences when the outcome is either ($Y, Y(d)$) and ($\mathcal{Y}, \mathcal{Y}(d)$).

\begin{example}[When $Y$ for each unit under treatment $d^{i}$ is scalar/vector]
	Suppose the potential outcome distribution is $\mathcal{N}(0,1)$. Then $Y$ for each unit under treatment $d^{i}$ is a realization from the distribution $\mathcal{N}(0,1)$.
\end{example}
\begin{example}[When $\mathcal{Y}$ for each unit under treatment $d^{i}$ is a distribution]
	We give a parametric example. Denote $\mu$ and $\sigma$ as the mean and standard deviation of a normal distribution. Suppose that $(\mu,\log\sigma)\sim\mathcal{N}\bigg(\begin{bmatrix}0\\0\end{bmatrix},\begin{bmatrix}0 & 1\\1 & 0\end{bmatrix}\bigg)$. If one realization of $(\mu,\log\sigma)$ is $(0.1,-0.5)$, then $(\mu,\sigma^{2})=(0.1,e^{-1})$. It means that for a given unit, there is a collection of observations that are drawn from $\mathcal{N}(0.1,e^{-1})$. If another realization of $(\mu,\log\sigma)$ is given (say $(-0.3,1)$), it means that a new unit would have a collection of observations that are drawn from $\mathcal{N}(-0.3,e^{2})$.
\end{example}
\noindent We further provide two examples to illustrate the differences between the ATE (average treatment effect) and the causal effect map. 
\begin{example}[When $Y$ for each unit under treatment $d^{i}$ is scalar/vector]
	Consider three units A, B, C, whose observations and potential outcomes are as follows:
	\begin{table}[!ht]
		\centering
		\begin{tabular}{cccc}
			\toprule
			Unit & $Y(0)$ & $Y(1)$ & $Y(1)-Y(0)$ \\ \hline
			A & 2 & 6 & 4 \\ 
			B & 1 & 5 & 4 \\ 
			C & 3 & 3 & 0 \\ 
			\bottomrule
		\end{tabular}
	\end{table}
	
	\noindent Then, the sampled average treatment effect (ATE) = (4+4+0)/3=8/3.  The sampled quantile treatment effect (QTE) at $\tau=0.5$ = median of $Y(1)$ - median of $Y(0)$ = 5 -  2 = 3.
	
\end{example}
\begin{example}[When $\mathcal{Y}$ for each unit under treatment $d^{i}$ is a distribution] Consider three units A, B, C, whose observations can be constituted as distributions (e.g., $[2, 0, 1]\sim \mathcal{N}(1, 1)$), respectively. The potential outcomes are as follows:
	\begin{table}[!ht]
		\centering
		\begin{tabular}{lll}
			\toprule
			\multicolumn{1}{c}{Unit} & \multicolumn{1}{c}{$\mathcal{Y}(0)$} & \multicolumn{1}{c}{$\mathcal{Y}(1)$}  \\ \hline
			\multicolumn{1}{c}{A} & $[2, 0, 1]\sim \mathcal{N}(1, 1)$ & $[6]\sim \mathcal{N}(5, 1)$  \\ 
			\multicolumn{1}{c}{B} & $[1, 3, 5]\sim \mathcal{N}(3, 1)$ & $[5, 10]\sim \mathcal{N}(8, 1)$  \\ 
			\multicolumn{1}{c}{C} & $[3, 4, 5]\sim \mathcal{N}(4, 1)$ & $[3, 6, 6, 6]\sim \mathcal{N}(6, 1)$  \\ 
			\bottomrule
		\end{tabular}
	\end{table}
	
	\noindent Note that for normal distributions with the same variance, the Wasserstein barycenter for these distributions is also normal distribution, with its mean equal to the averaged mean of these distributions since the Wasserstein barycenter is the distribution that has the smallest Wasserstein distance of these distributions. For example, the Wasserstein barycenter for $\mathcal{N}(a, s^2)$, $\mathcal{N}(b, s^2)$, $\mathcal{N}(c, s^2)$ is $\mathcal{N}\bigg(\frac{a+b+c}{3}, s^2\bigg)$ with the following CDF:
		\begin{equation*}
			{\small
				\begin{aligned}
					\int_{-\infty}^x \frac{1}{\sqrt{2\pi}s}\exp\left(\frac{(z-\frac{a+b+c}{3})^2}{2s^2}\right)dz.
				\end{aligned}
			}
		\end{equation*}
		Thus, in this example, the sampled barycenter for $\mathcal{Y}(0)$ is 
		\begin{equation*}
			{\small
				\begin{aligned}
					\mathrm{Erf}(x)=\int_{-\infty}^x \frac{1}{\sqrt{2\pi}}\exp\left(\frac{(z-\frac{8}{3})^2}{2}\right)dz
				\end{aligned}
			}
		\end{equation*}
		and the sampled barycenter for $\mathcal{Y}(1)$ is 
		\begin{equation*}
			{\small
				\begin{aligned}
					\mathrm{Erf}(x)=\int_{-\infty}^x \frac{1}{\sqrt{2\pi}}\exp\left(\frac{(z-\frac{19}{3})^2}{2}\right)dz.
				\end{aligned}
			}
		\end{equation*}
		Thus, the causal effect map between $\mathcal{Y}(0)$ and $\mathcal{Y}(1)$ is 
		\begin{equation*}
			{\small
				\begin{aligned}
					\mathrm{Erf}^{-1}(x)-\mathrm{Erf}^{-1}(x).
				\end{aligned}
			}
		\end{equation*}

	\noindent Note that in real cases, the distributions for each unit are more complex than the normal distribution, and they even may not follow any known distributions.
	
\end{example}

\section{Consistency of DR and IPW estimators}\label{sec:consistency studies}
In this section, we are going to study the asymptotic properties of the DR and IPW estimators. The studies require the following Lemmas:
\begin{lemma}\label{lemma:equivalent of DR}
	Let $m_{d^{i}}(\mathbf{X})=\mathbb{E}[\mathcal{Y}^{-1}|D=d^{i},\mathbf{X}]$. We have
	\begin{equation*}
		\begin{aligned}
			\bigtriangleup_{d^{i}}(\cdot)=\mathbb{E}[m_{d^{i}}(\mathbf{X})]
		\end{aligned}
	\end{equation*}
\end{lemma}
\begin{proof}
	See the proof of Assertion 2 in the proof of Proposition \ref{prop:causal effect expectation form}.
\end{proof}
\begin{lemma}\label{lemma:equivalent of IPW}
	Let $\pi_{d^{i}}(\mathbf{X})=\mathbb{P}\{D=d^{i}|\mathbf{X}\}$. We have
	\begin{equation*}
		\begin{aligned}
			\bigtriangleup_{d^{i}}(\cdot)=\mathbb{E}\bigg[\frac{\mathbf{1}_{\{D=d^{i}\}}}{\pi_{d^{i}}(\mathbf{X})}\mathcal{Y}^{-1}\bigg]
		\end{aligned}
	\end{equation*}
\end{lemma}
\begin{proof}
	From the proof of Assertion 1 in the proof of Proposition \ref{prop:causal effect expectation form}, we notice that $\bigtriangleup_{d^{i}}(\cdot)=\mathbb{E}[\mathcal{Y}(d^{i})^{-1}]$. We only need to show that
	\begin{equation*}
		\begin{aligned}
			\mathbb{E}[\mathcal{Y}(d^{i})^{-1}]=\mathbb{E}\bigg[\frac{\mathbf{1}_{\{D=d^{i}\}}}{\pi_{d^{i}}(\mathbf{X})}\mathcal{Y}^{-1}\bigg].
		\end{aligned}
	\end{equation*}
	Now, we have
	\begin{equation*}
		\begin{aligned}
			&\mathbb{E}\bigg[\frac{\mathbf{1}_{\{D=d^{i}\}}}{\pi_{d^{i}}(\mathbf{X})}\mathcal{Y}^{-1}\bigg]=\mathbb{E}\bigg[\mathbb{E}\bigg[\frac{\mathbf{1}_{\{D=d^{i}\}}}{\pi_{d^{i}}(\mathbf{X})}\mathcal{Y}^{-1}|\mathbf{X}\bigg]\bigg]\\
			=&\mathbb{E}\bigg[\frac{1}{\pi_{d^{i}}(\mathbf{X})}\mathbb{E}\bigg[\mathbf{1}_{\{D=d^{i}\}}\mathcal{Y}^{-1}|\mathbf{X}\bigg]\bigg]\\
			=&\mathbb{E}\bigg[\frac{1}{\pi_{d^{i}}(\mathbf{X})}\underset{1\leq j\leq r}{\sum}\mathbb{E}[\mathbf{1}_{\{D=d^{i}\}}\mathcal{Y}^{-1}|D=d^{j},\mathbf{X}]\mathbb{P}\{D=d^{j}|\mathbf{X}\}\bigg]\\
			=&\mathbb{E}\bigg[\frac{1}{\pi_{d^{i}}(\mathbf{X})}\mathbb{E}[\mathcal{Y}^{-1}|D=d^{i},\mathbf{X}]\mathbb{P}\{D=d^{i}|\mathbf{X}\}\bigg]\\=&\mathbb{E}[\mathbb{E}[\mathcal{Y}^{-1}|D=d^{i},\mathbf{X}]]\\
			\overset{\ast}{=}&\mathbb{E}[\mathbb{E}[\mathcal{Y}(d^{i})^{-1}|D=d^{i},\mathbf{X}]]\overset{\star}{=}\mathbb{E}[\mathbb{E}[\mathcal{Y}(d^{i})^{-1}|\mathbf{X}]]\\=&\mathbb{E}[\mathcal{Y}(d^{i})^{-1}].
		\end{aligned}
	\end{equation*}
	Here, $\ast$ follows from Consistency Assumption, and $\star$ follows from Ignorability Assumption.
\end{proof}
From Lemmas \ref{lemma:equivalent of DR} and \ref{lemma:equivalent of IPW}, we can study the asymptotic properties of $\hat{\bigtriangleup}_{d^i;DR}$ and $\hat{\bigtriangleup}_{d^i;IPW}$. Recall the estimators $\hat{\bigtriangleup}_{d^i;DR}$ and $\hat{\bigtriangleup}_{d^i;IPW}$ here:

\noindent \underline{\textbf{DR estimator $\hat{\bigtriangleup}_{d^i;DR}$}}:
\begin{equation}
	\begin{aligned}\label{eqt:cross-fit estimator final DR}
		\hat{\bigtriangleup}_{d^i;DR}=\underset{k=1}{\overset{K}{\sum}}\frac{N_{k}}{N}\frac{1}{N_{k}}\underset{s\in\mathcal{D}_{k}}{\overset{}{\sum}}\hat{m}_{d^i}^{k}(\mathbf{X}_{s}). 
	\end{aligned}
\end{equation} 
\underline{\textbf{IPW estimator $\hat{\bigtriangleup}_{d^i;IPW}$}}:
\begin{equation}
	\begin{aligned}\label{eqt:cross-fit estimator final IPW}
		\hat{\bigtriangleup}_{d^i;IPW}=\underset{k=1}{\overset{K}{\sum}}\frac{N_{k}}{N}\frac{1}{N_{k}}\underset{s\in\mathcal{D}_{k}}{\overset{}{\sum}}\frac{\mathbf{1}_{\{D_{s}=d^{i}\}}}{\hat{\pi}^k_{d^i}(\mathbf{X}_{s})}\mathcal{Y}_{s}^{-1}. 
	\end{aligned}
\end{equation}  
We first study the asymptotic property of $\hat{\bigtriangleup}_{d^i;DR}$. Write
\begin{equation*}
	\begin{aligned}
		\hat{\bigtriangleup}_{d^i;DR}=&\underset{k=1}{\overset{K}{\sum}}\frac{N_{k}}{N}\frac{1}{N_{k}}\underset{s\in\mathcal{D}_{k}}{\overset{}{\sum}}\hat{m}_{d^i}^{k}(\mathbf{X}_{s})\\
		=&\underbrace{\underset{k=1}{\overset{K}{\sum}}\frac{N_{k}}{N}\frac{1}{N_{k}}\underset{s\in\mathcal{D}_{k}}{\overset{}{\sum}}[\hat{m}_{d^i}^{k}(\mathbf{X}_{s})-m_{d^i}(\mathbf{X}_{s})]}_{\text{I}}+\underbrace{\underset{k=1}{\overset{K}{\sum}}\frac{N_{k}}{N}\frac{1}{N_{k}}\underset{s\in\mathcal{D}_{k}}{\overset{}{\sum}}m_{d^i}(\mathbf{X}_{s})}_{\text{II}}.
	\end{aligned}
\end{equation*}
$\text{II}$ is the sample averaging of $\mathbb{E}[m_{d^{i}}(\mathbf{X})]$, so $\text{II}$ converges to $\mathbb{E}[m_{d^{i}}(\mathbf{X})]$ in probability. By Lemma \ref{lemma:equivalent of DR}, $\text{II}$ converges to $\bigtriangleup_{d^{i}}(\cdot)$ in probability. Hence, the consistency of the estimator $\hat{\bigtriangleup}_{d^i;DR}$ depends on the convergence of $\hat{m}_{d^i}^{k}(\mathbf{X}_{s})$ on $m_{d^i}(\mathbf{X}_{s})$.

We then study the asymptotic property of $\hat{\bigtriangleup}_{d^i;IPW}$. Similarly, we can write
\begin{equation*}
	\begin{aligned}
		\hat{\bigtriangleup}_{d^i;IPW}=&\underset{k=1}{\overset{K}{\sum}}\frac{N_{k}}{N}\frac{1}{N_{k}}\underset{s\in\mathcal{D}_{k}}{\overset{}{\sum}}\frac{\mathbf{1}_{\{D_{s}=d^{i}\}}}{\hat{\pi}^k_{d^i}(\mathbf{X}_{s})}\mathcal{Y}_{s}^{-1}\\
		=&\underbrace{\underset{k=1}{\overset{K}{\sum}}\frac{N_{k}}{N}\frac{1}{N_{k}}\underset{s\in\mathcal{D}_{k}}{\overset{}{\sum}}\biggl[\frac{1}{\hat{\pi}^k_{d^i}(\mathbf{X}_{s})}-\frac{1}{\pi_{d^i}(\mathbf{X}_{s})}\biggl]\mathbf{1}_{\{D_{s}=d^{i}\}}\mathcal{Y}_{s}^{-1}}_{\text{I}}+\underbrace{\underset{k=1}{\overset{K}{\sum}}\frac{N_{k}}{N}\frac{1}{N_{k}}\underset{s\in\mathcal{D}_{k}}{\overset{}{\sum}}\frac{\mathbf{1}_{\{D_{s}=d^{i}\}}}{\pi_{d^i}(\mathbf{X}_{s})}\mathcal{Y}_{s}^{-1}}_{\text{II}}.
	\end{aligned}
\end{equation*}
Again, $\text{II}$ is the sample averaging of $\mathbb{E}\bigg[\frac{\mathbf{1}_{\{D=d^{i}\}}}{\pi_{d^{i}}(\mathbf{X})}\mathcal{Y}^{-1}\bigg]$, so $\text{II}$ converges to $\mathbb{E}\bigg[\frac{\mathbf{1}_{\{D=d^{i}\}}}{\pi_{d^{i}}(\mathbf{X})}\mathcal{Y}^{-1}\bigg]$ in probability. By Lemma \ref{lemma:equivalent of IPW}, $\text{II}$ converges to $\bigtriangleup_{d^{i}}(\cdot)$ in probability. Hence, the consistency of the estimator $\hat{\bigtriangleup}_{d^i;DR}$ depends on the convergence of $\hat{\pi}_{d^i}^{k}(\mathbf{X}_{s})$ on $\pi_{d^i}(\mathbf{X}_{s})$.

\section{Missing Proofs}\label{sec:missing proofs}

\subsection{Proof of Proposition \ref{prop:causal effect expectation form}}
\begin{proof}[Proof of Proposition \ref{prop:causal effect expectation form}]
	Proof of Assertion 1: If we can prove that $\mathbb{E}\big[\mathcal{Y}(d^{i})^{-1}\big]=\mu_{d^i}^{-1}$, then we have $\bigtriangleup_{d^i}=\mu_{d^i}^{-1}=\mathbb{E}\big[\mathcal{Y}(d^{i})^{-1}\big]$. Let $\mathcal{Q}$ be the set containing all the left-continuous non-decreasing functions on $(0,1)$. If we view $\mathcal{Q}$ as a subspace of $L^{2}([0,1])$, then it is isometric to $\mathcal{W}_{2}(\mathcal{I})$ \citep{panaretos2020invitation}. Indeed, $\mu_{d^i}=\underset{\nu\in\mathcal{W}_{2}(\mathcal{I})}{\arg\min}\mathbb{E}\big[\mathbb{D}_{2}(\mathcal{Y}(d^{i}),\nu)^{2}\big]\overset{\bullet}{=}\underset{\nu\in\mathcal{Q}}{\arg\min}\;\mathbb{E}\big[\int_{0}^{1}|\mathcal{Y}(d^{i})^{-1}(t)-\nu^{-1}(t)|^{2}dt\big]$. Here, $\overset{\bullet}{=}$ follows from Theorem 2.18 of \cite{villani2021topics}. Since we can interchange the integral sign $\int$ and $\mathbb{E}$, we notice that $\mathbb{E}\big[\int_{0}^{1}|\mathcal{Y}(d^{i})^{-1}(t)-\nu^{-1}(t)|^{2}dt\big]=\int_{0}^{1}\mathbb{E}\big[|\mathcal{Y}(d^{i})^{-1}(t)-\nu^{-1}(t)|^{2}\big]dt=\int_{0}^{1}(\mathbb{E}\big[\mathcal{Y}(d^{i})^{-1}(t)\big]-\nu^{-1}(t))^{2}dt+\int_{0}^{1}\mathbb{E}[(\mathbb{E}\big[\mathcal{Y}(d^{i})^{-1}(t)\big]-\mathcal{Y}(d^{i})^{-1}(t))^{2}]dt$, and $\mathbb{E}\big[\int_{0}^{1}|\mathcal{Y}(d^{i})^{-1}(t)-\nu^{-1}(t)|^{2}dt\big]$ attains its minimum when $\nu^{-1}(t)=\mathbb{E}\big[\mathcal{Y}(d^{i})^{-1}(t)\big]$. We can therefore conclude that $\mu_{d^i}=\big(\mathbb{E}\big[\mathcal{Y}(d^{i})^{-1}(t)\big]\big)^{-1}$.
	
	Proof of Assertion 2: From Proposition \ref{prop:causal effect expectation form}, $\bigtriangleup_{d^i}(\cdot)=\mathbb{E}\big[\mathcal{Y}(d^{i})^{-1}\big]$. Consistency assumption assures that $\mathcal{Y}^{-1}=\mathcal{Y}(d^{i})^{-1}$ since $\mathcal{Y}=\mathcal{Y}(d^{i})$ when $D=d^{i}$. Consequently, we have $\mathbb{E}\big[\mathcal{Y}(d^{i})^{-1}\big]=\mathbb{E}\big[\mathbb{E}\big[\mathcal{Y}(d^{i})^{-1}|\mathbf{X}\big]\big]\overset{\star}{=}\mathbb{E}\big[\mathbb{E}\big[\mathcal{Y}(d^{i})^{-1}|D=d^{i},\mathbf{X}\big]\big]\overset{\ast}{=}\mathbb{E}\big[\mathbb{E}\big[\mathcal{Y}^{-1}|D=d^{i},\mathbf{X}\big]\big]$. $\ast$ follows from Ignorability Assumption while $\star$ follows from Consistency Assumption. Thus, we conclude that $\bigtriangleup_{d^i}(\cdot)$ is identified.
\end{proof}

Before starting the proof of Theorem \ref{thm:root N consistent}, we introduce some notations that are useful in the proof of Theorem \ref{thm:root N consistent}.

We put a $\hat{}$ on top of a random variable which represents the estimate of the random variable. For example, we use $\hat{\mathcal{Y}}$ and $\hat{\pi}_{d^i}$ to represent the estimate of $\mathcal{Y}$ and $\pi_{d^i}$. Since we spit the observed data into $K$ disjoint union sets $\mathcal{D}_{1},\cdots,\mathcal{D}_{K}$, we also use $\hat{\pi}^{k}_{d^i}$ to represent the estimate $\pi_{d^i}$ based on the set $\mathcal{D}_{-k}$ and evaluate using the set $\mathcal{D}_{k}$.

We use $m_{d^i}(\mathbf{X})=\mathbb{E}[\mathcal{Y}(d^{i})^{-1}|\mathbf{X}]=\mathbb{E}[\mathcal{Y}^{-1}|D=d^{i},\mathbf{X}]$. Further, denote $\tilde{m}_{d^i}^{k}$ as the estimate using true $\mathcal{Y}$ based on the set $\mathcal{D}_{-k}$ and evaluate using the set $\mathcal{D}_{k}$. Define an inner product space with inner product $\langle h_{1},h_{2}\rangle=\int h_{1}(t)h_{2}(t)dt$ such that $\|h\|^{2}=\int h^{2}(t)dt$. Further, given two functions $h_{1}(\cdot,\mathbf{X})$ and $h_{2}(\cdot,\mathbf{X})$, we define $\vvvert h_{1}-h_{2}\vvvert^{2}$ such that $\vvvert h_{1}-h_{2}\vvvert^{2}=\int \|h_{1}(\cdot,\mathbf{x})-h_{2}(\cdot,\mathbf{x})\|^{2}dF_{\mathbf{X}}(\mathbf{x})=\int \int |h_{1}(t,\mathbf{x})-h_{2}(t,\mathbf{x})|^{2} dt\;dF_{\mathbf{X}}(\mathbf{x})=\mathbb{E}[\|h_{1}-h_{2}\|^{2}]$, where $F_{\mathbf{X}}(\mathbf{x})$ is the cumulative distribution of $\mathbf{X}$. 
Besides, let $\rho_{\pi}^{4}=\mathbb{E}[|\hat{\pi}_{d^i}(\mathbf{X})-\pi_{d^i}(\mathbf{X})|^{4}]$;
$\rho_{m}^{4}=\max\{\vvvert\tilde{m}_{d^i}-m_{d^i}\vvvert^{4},\;1\leq i\leq r\}$. The expectation $\mathbb{E}$ are conditional expectation conditioning on the estimated nuisance functions.
We also assume that $\mathbb{E}[\|m_{d^i}(\mathbf{X})\|^{4}]$, $\mathbb{E}[\|\mathcal{Y}^{-1}\|^{4}]$, and $\mathbb{E}[\|\mathcal{Y}(d^{i})^{-1}\|_{\lambda}^{4}]<\infty$ $\forall\;1\leq i\leq r$.

\noindent Next, we state the convergence assumptions which are needed in the proof of Theorem \ref{thm:root N consistent}.

\begin{convassumption}\label{ass:assumption1}
	The estimates $\hat{\mathcal{Y}}_{1},\cdots,\hat{\mathcal{Y}}_{N}$ are independent of each other. Furthermore, there are two sequences of constants $\alpha_{N}=o(1)$ and $\nu_{N}=o(1)$ such that
	\begin{equation*}
		\begin{aligned}
			\underset{1\leq s\leq N}{\sup}\underset{v\in\mathcal{W}(\mathcal{I})}{\sup}\mathbb{E}[\mathbb{D}_{2}^{2}(\hat{\mathcal{Y}}_{s},\mathcal{Y}_{s})|\mathcal{Y}_{s}=v]&=O(\alpha_{N}^{2})\\
			\underset{1\leq s\leq N}{\sup}\underset{v\in\mathcal{W}(\mathcal{I})}{\sup}\mathbb{V}[\mathbb{D}_{2}^{2}(\hat{\mathcal{Y}}_{s},\mathcal{Y}_{s})|\mathcal{Y}_{s}=v]&=O(\nu_{N}^{4}).
		\end{aligned}
	\end{equation*}
\end{convassumption}

\begin{convassumption}\ \label{ass:assumption3}
	\begin{enumerate}
		\item There exists such that
		\begin{equation*}
			\begin{aligned}
				\mathbb{P}\{\epsilon<\hat{\pi}_{d^i}^{k}<1-\epsilon,\;\forall\;x\}=1.
			\end{aligned}
		\end{equation*}
		\item The outcome regression and propensity score estimates converge: $\forall\;i\in\{1,\cdots,r\}$ and $\forall\;1\leq k\leq K$, we have
		\begin{equation*}
			\begin{aligned}
				\underset{x}{\sup}\|\tilde{m}_{d^i}^k-m_{d^i}\|&=o_{P}(1)\quad\textit{and} \quad
				\underset{x}{\sup}\|\hat{\pi}_{d^i}^{k}-\pi_{d^i}\|&=o_{P}(1).
			\end{aligned}
		\end{equation*}
		
	\end{enumerate}
\end{convassumption}

\begin{convassumption}\label{ass:assumption5}
	$\vvvert \hat{m}_{d^i}^k-\tilde{m}_{d^i}^k\vvvert=O_{P}(N^{-1}+\alpha_{N}^{2}+\nu_{N}^{2})$ $\forall\;i\in\{1,\cdots,r\}$ and $1\leq k\leq K$.
\end{convassumption}
\begin{convassumption}\label{ass:assumption6}
	There exist constants $c_{1}$ and $c_{2}$ such that $0<c_{1}\leq \frac{N_{k}}{N}\leq c_{2}<1$ for all $N$ and $1\leq k\leq K$.
\end{convassumption}
\noindent To prove Theorem 1, we require all the Convergence Assumptions. Further, we need to assume that the two sequences $\alpha_{N}$ and $\nu_{N}$ in Convergence Assumption \ref{ass:assumption1} satisfy $\alpha_{N}=o(N^{-\frac{1}{2}})$ and $\nu_{N}=o(N^{-\frac{1}{2}})$. Note that $\alpha_{N}=o(N^{-\frac{1}{2}})$ and $\nu_{N}=o(N^{-\frac{1}{2}})$ holds imply that $\alpha_{N}=o(1)$ and $\nu_{N}=o(1)$ automatically.

\subsection{Proof of Theorem \ref{thm:root N consistent}}

The proof requires two Lemmas.
\begin{lemma}\label{lemma:1Lemma}
	For $G_{1},\;G_{2}\in\mathcal{W}_{2}(\mathcal{I})$, we have $\|G_{1}-G_{2}\|=\mathbb{D}_{2}(G_{1},G_{2})$.
\end{lemma}
\begin{lemma}\label{lemma:2Lemma}
	Under Convergence Assumption \ref{ass:assumption1}, we have $\frac{1}{N}\underset{s=1}{\overset{N}{\sum}}\|\hat{\mathcal{Y}}_{s}^{-1}-\mathcal{Y}_{s}^{-1}\|^{2}=O_{P}(\alpha_{N}^{2}+\nu_{N}^{2})$.
\end{lemma}
\noindent The proofs of Lemmas \ref{lemma:1Lemma} and \ref{lemma:2Lemma} can be found in \cite{lin2023causal}. Now, we are ready to prove Theorem \ref{thm:root N consistent}.

\begin{proof}[Proof of Theorem \ref{thm:root N consistent}]
	In the following proof, we assume that $K=2$. The general case is similar. For simplicity, we define four operators $\mathbb{P}_{N}$, $\mathbb{P}_{N_{k}}$, $\mathbb{E}_{N}$, and $\mathbb{E}_{N_{k}}$ such that given a random quantity $\mathcal{O}$, $\mathbb{P}_{N}\mathcal{O}=\frac{1}{N}\underset{s=1}{\overset{N}{\sum}}\mathcal{O}_{s}$, $\mathbb{P}_{N_{k}}=\frac{1}{N_{k}}\underset{s\in\mathcal{D}_{k}}{\overset{}{\sum}}\mathcal{O}_{s}$, $\mathbb{E}_{N}=\frac{1}{N}\underset{s=1}{\overset{N}{\sum}}\mathbb{E}[\mathcal{O}_{s}]$, and $\mathbb{E}_{N_{k}}=\frac{1}{N_{k}}\underset{s\in\mathcal{D}_{k}}{\overset{}{\sum}}\mathbb{E}[\mathcal{O}_{s}]$. Given the distributions $\lambda$. Define $\mathcal{L}\lambda=\lambda^{-1}$. Let $Z_{s}=\mathcal{L}\mathcal{Y}_{s}$, and if the $s^{\mathrm{th}}$ subject belongs to the $k$ partition, then $\hat{Z}_{s}=\mathcal{L}\hat{\mathcal{Y}}_{s}$ and $R_{s}=\hat{Z}_{s}-Z_{s}$. Define $D_{d^i}^k(x)=\hat{m}_{d^i}^k(x)-\tilde{m}_{d^i}^k(x)$.  Under the causal assumptions, we have
	\begin{equation*}
		\begin{aligned}
			\bigtriangleup_{d^i}=\mathbb{E}\bigg[\frac{\mathbf{1}_{\{D=d^{i}\}}\mathcal{L}\mathcal{Y}}{\pi_{d^i}(\mathbf{X})}-\bigg(\frac{\mathbf{1}_{\{D=d^{i}\}}}{\pi_{d^i}(\mathbf{X})}-1\bigg)m_{d^i}(\mathbf{X})\bigg].
		\end{aligned}
	\end{equation*}
	Denote the corresponding sampled version using $\mathcal{D}_{k}$ as
	\begin{equation*}
		{\small
			\begin{aligned}
				\hat{\bigtriangleup}_{d^i;DML}^{k}=\mathbb{P}_{N_{k}}\bigg[\frac{\mathbf{1}_{\{D=d^{i}\}}\mathcal{L}\hat{\mathcal{Y}}}{\hat{\pi}^{k}_{d^i}(\mathbf{X})}-\bigg(\frac{\mathbf{1}_{\{D=d^{i}\}}}{\hat{\pi}_{d^i}^{k}(\mathbf{X})}-1\bigg)\hat{m}_{d^i}^k(\mathbf{X})\bigg].
			\end{aligned}
		}
	\end{equation*}
	As a result, we have the cross-fitting estimator $\hat{\bigtriangleup}_{d^i;DML}$ such that
	\begin{equation*}
		\begin{aligned}
			\hat{\bigtriangleup}_{d^i;DML}&=\underset{k=1}{\overset{2}{\sum}}\frac{N_{k}}{N}\hat{\bigtriangleup}_{d^i;DML}^{k}\\&=\frac{1}{N}(N_{1}\hat{\bigtriangleup}_{d^i;DML}^{1}+N_{2}\hat{\bigtriangleup}_{d^i;DML}^{2}).
		\end{aligned}
	\end{equation*}
	
	Next, we consider the difference $\hat{\bigtriangleup}_{d^i;DML}-\bigtriangleup_{d^i}^{\lambda}$. Indeed, we have
	\begin{equation*}
		{\small
			\begin{aligned}
				&\frac{1}{N}(N_{1}\hat{\bigtriangleup}_{d^i;DML}^{1}+N_{2}\hat{\bigtriangleup}_{d^i;DML}^{2})-\bigtriangleup_{d^i}=\frac{1}{N}\underset{k=1,2}{\overset{}{\sum}}N_{k}\mathcal{A}_{k}-\bigtriangleup_{d^i},
			\end{aligned}
		}
	\end{equation*}
	where
	\begin{equation*}
		\begin{aligned}
			\mathcal{A}_{k}=\mathbb{P}_{N_{k}}\bigg[\frac{\mathbf{1}_{\{D=d^{i}\}}Z+\mathbf{1}_{\{D=d^{i}\}}R}{\hat{\pi}^{k}_{d^i}(\mathbf{X})}
			-\bigg(\frac{\mathbf{1}_{\{D=d^{i}\}}}{\hat{\pi}^{k}_{d^i}(\mathbf{X})}-1\bigg)(\tilde{m}_{d^i}^k(\mathbf{X})+D_{d^i}^k(\mathbf{X}))\bigg].
		\end{aligned}
	\end{equation*}
	
	We then decompose $\frac{1}{N}\underset{k=1,2}{\overset{}{\sum}}N_{k}\mathcal{A}_{k}-\bigtriangleup_{i}$ into the sum of five quantities as follows:
	\begin{equation*}
		\begin{aligned}
			\frac{1}{N}\underset{k=1,2}{\overset{}{\sum}}N_{k}(\text{I}+\text{II}+\text{III}+\text{IV}+\text{V})\end{aligned}
	\end{equation*}
	where
	\begin{equation*}
		\begin{aligned}
			\text{I}&=(\mathbb{P}_{N_{k}}-\mathbb{E}_{N_{k}})\bigg[\frac{\mathbf{1}_{\{D=d^{i}\}}(Z-\tilde{m}_{d^i}^k(\mathbf{X}))}{\hat{\pi}_{d^i}^{k}(\mathbf{X})}+\tilde{m}_{d^i}^k(\mathbf{X})-\frac{\mathbf{1}_{\{D=d^{i}\}}(Z-m_{d^i}(\mathbf{X}))}{\pi_{d^i}(\mathbf{X})}-m_{d^i}(\mathbf{X})\bigg]\\
			\text{II}&=(\mathbb{P}_{N_{k}}-\mathbb{E}_{N_{k}})\bigg[\frac{\mathbf{1}_{\{D=d^{i}\}}(Z-m_{d^i}(\mathbf{X}))}{\pi_{d^i}(\mathbf{X})}+m_{d^i}(\mathbf{X})\bigg]\\
			\text{III}&=\mathbb{E}_{N_{k}}\bigg[\frac{(\tilde{m}_{d^i}^k(\mathbf{X})-m_{d^i}(\mathbf{X}))(\hat{\pi}^{k}_{d^i}(\mathbf{X})-\mathbf{1}_{\{D=d^{i}\}})}{\hat{\pi}^{k}_{d^i}(\mathbf{X})}\bigg]\\
			\text{IV}&=\mathbb{P}_{N_{k}}\bigg\{\bigg(1-\frac{\mathbf{1}_{\{D=d^{i}\}}}{\hat{\pi}^{k}_{d^i}(\mathbf{X})}\bigg)D_{d^i}^k(\mathbf{X})\bigg\}\\
			\text{V}&=\mathbb{P}_{N_{k}}\bigg\{\frac{\mathbf{1}_{\{D=d^{i}\}}R}{\hat{\pi}^{k}_{d^i}(\mathbf{X})}\bigg\}.
		\end{aligned}
	\end{equation*}
	\textit{Proof of Theorem \ref{thm:root N consistent}.\ref{thm:root N consistent1}}: The proof follows from the bounds of I - V (see more details in the sequel). 
	
	\noindent\textit{Proof of Theorem \ref{thm:root N consistent}.\ref{thm:root N consistent2}}: The proof follows from the Central Limit Theorem and Slutsky's Lemma after we get the bounds of I - V.
	
	\noindent In the sequel, we bound I - V subsequently.
	
	\noindent\underline{Boundness of I:} Let 
	\begin{equation*}
		{\small
			\begin{aligned}
				G(D,\mathbf{X},Z)&=\frac{\mathbf{1}_{\{D=d^{i}\}}\{Z-\tilde{m}_{d^i}^k(\mathbf{X})\}}{\hat{\pi}^{k}_{d^i}(\mathbf{X})}+\tilde{m}_{d^i}^k(\mathbf{X})-m_{d^i}(\mathbf{X})\\
				H_{1}(D,\mathbf{X},Z)&=\frac{\mathbf{1}_{\{D=d^{i}\}}Z\{\pi_{d^i}(\mathbf{X})-\hat{\pi}^{k}_{d^i}(\mathbf{X})\}}{\hat{\pi}^{k}_{d^i}(\mathbf{X})\pi_{d^i}(\mathbf{X})}\\
				H_{2}(D,\mathbf{X},Z)&=\frac{\mathbf{1}_{\{D=d^{i}\}}\{\hat{\pi}^{k}_{d^i}(\mathbf{X})m_{d^i}(\mathbf{X})-\pi_{d^i}(\mathbf{X})\tilde{m}_{d^i}^k(\mathbf{X})\}}{\hat{\pi}^{k}_{d^i}(\mathbf{X})\pi_{d^i}(\mathbf{X})}\\
				H_{3}(D,\mathbf{X},Z)&=\tilde{m}_{d^i}^k(\mathbf{X})-m_{d^i}(\mathbf{X})\\
				H(D,\mathbf{X},Z)&=H_{1}(D,\mathbf{X},Z)+H_{2}(D,\mathbf{X},Z)+H_{3}(D,\mathbf{X},Z).
			\end{aligned}
		}
	\end{equation*}
	We consider $\mathbb{E}[\|\text{I}\|^{2}]$. Note that 
	\begin{equation*}
		{\small
			\begin{aligned}
				H(D,\mathbf{X},Z)&=\bigg[\frac{\mathbf{1}_{\{D=d^{i}\}}(Z-\tilde{m}_{d^i}^k(\mathbf{X}))}{\hat{\pi}^{k}_{d^i}(\mathbf{X})}+\tilde{m}_{d^i}^k(\mathbf{X})-\\
				&\quad\quad\quad\frac{\mathbf{1}_{\{D=d^{i}\}}(Z-m_{d^i}(\mathbf{X}))}{\pi_{d^i}(\mathbf{X})}-m_{d^i}(\mathbf{X})\bigg].
			\end{aligned}
		}
	\end{equation*}
	Hence, we have
	\begin{equation*}
		\begin{aligned}
			&\mathbb{E}[\|\text{I}\|^{2}]=\mathbb{E}[\|(\mathbb{P}_{N_{k}}-\mathbb{E}_{N_{k}})H(D,\mathbf{X},Z)\|^{2}].
		\end{aligned}
	\end{equation*}
	We simplify the quantity $\mathbb{E}[\|(\mathbb{P}_{N_{k}}-\mathbb{E}_{N_{k}})H(D,\mathbf{X},Z)\|^{2}]$. Indeed, we have
	\begin{equation*}
		{\small
			\begin{aligned}
				&\mathbb{E}[\|(\mathbb{P}_{N_{k}}-\mathbb{E}_{N_{k}})H(D,\mathbf{X},Z)\|^{2}]\\&=\frac{1}{N_{k}^{2}}\mathbb{E}[\|\underset{s\in\mathcal{D}_{k}}{\sum}\{H(D_{s},\mathbf{X}_{s},Z_{s})-\mathbb{E}[H(D_{s},\mathbf{X}_{s},Z_{s})]\}\|^{2}]\\
				&=\frac{1}{N_{k}^{2}}\underset{s\in\mathcal{D}_{k}}{\sum}\mathbb{E}[\| H(D_{s},\mathbf{X}_{s},Z_{s})-\mathbb{E}[H(D_{s},\mathbf{X}_{s},Z_{s})]\|^{2}]+\frac{1}{N_{k}^{2}}\underset{\substack{s,\bar{s}\in\mathcal{D}_{k}\\s\neq \bar{s}}}{\sum}C_{s\bar{s}}:=\text{I}_{1}+\text{I}_{2},
			\end{aligned}
		}
	\end{equation*}
	where $C_{s\bar{s}}=\mathbb{E}[\langle H_{s}-\mathbb{E}[H_{s}],H_{\bar{s}}-\mathbb{E}[H_{\bar{s}}]\rangle]$ and $H_{s}=H(D_{s},\mathbf{X}_{s},Z_{s})$. Consider the term $\text{I}_{1}$. We have
	\begin{equation*}
		{\small
			\begin{aligned}
				\text{I}_{1}&\lesssim\frac{1}{N_{k}^{2}}\underset{s\in\mathcal{D}_{k}}{\sum}\mathbb{E}[\| H(D_{s},\mathbf{X}_{s},Z_{s})\|^{2}]\\
				&\leq\frac{1}{N_{k}^{2}}\underset{s\in\mathcal{D}_{k}}{\sum}\mathbb{E}[\|H_{1}(D_{s},\mathbf{X}_{s},Z_{s})\|^{2}]
				+\frac{1}{N_{k}^{2}}\underset{s\in\mathcal{D}_{k}}{\sum}\mathbb{E}[\|H_{2}(D_{s},\mathbf{X}_{s},Z_{s})\|^{2}]+\frac{1}{N_{k}^{2}}\underset{s\in\mathcal{D}_{k}}{\sum}\mathbb{E}[\|H_{3}(D_{s},\mathbf{X}_{s},Z_{s})\|^{2}].
			\end{aligned}
		}
	\end{equation*}
	Consider the bound $\mathbb{E}[\|H_{1}(D_{s},\mathbf{X}_{s},Z_{s})\|^{2}]$, Using the assumption that there exists $\epsilon>0$ such that $\hat{\pi}^{k}_{d^i}(\mathbf{X})$ and $\pi_{d^i}(\mathbf{X})$ are bounded below by $\epsilon\leq\hat{\pi}^{k}_{d^i}(\mathbf{X}),\pi_{d^i}(\mathbf{X})\leq 1-\epsilon$, we have
	\begin{equation*}
		{\small
			\begin{aligned}
				&\mathbb{E}[\|H_{1}(D_{s},\mathbf{X}_{s},Z_{s})\|^{2}]=\mathbb{E}\bigg[\bigg\|\frac{\mathbf{1}_{\{D=d^{i}\}}Z\{\pi_{d^i}(\mathbf{X})-\hat{\pi}^{k}_{d^i}(\mathbf{X})\}}{\hat{\pi}^{k}_{d^i}(\mathbf{X})\pi_{d^i}(\mathbf{X})}\bigg\|^{2}\bigg]\\
				\lesssim&\mathbb{E}[\|Z\{\pi_{d^i}(\mathbf{X})-\hat{\pi}^{k}_{d^i}(\mathbf{X})\}\|^{2}]=\mathbb{E}[|\pi_{d^i}(\mathbf{X})-\hat{\pi}^{k}_{d^i}(\mathbf{X})|^{2}\|Z\|^{2}]\\
				\leq&\big(\mathbb{E}[|\pi_{d^i}(\mathbf{X})-\hat{\pi}^{k}_{d^i}(\mathbf{X})|^{4}]\big)^{\frac{1}{2}}\big(\mathbb{E}[\|Z\|^{4}]\big)^{\frac{1}{2}}\lesssim \rho_{\pi}^{2}.
			\end{aligned}
		}
	\end{equation*}
	Next, consider the bound $\mathbb{E}[\|H_{2}(D_{s},\mathbf{X}_{s},Z_{s})\|^{2}]$. Again, using the assumption that there exists $\epsilon>0$ such that $\hat{\pi}^{k}_{d^i}(\mathbf{X})$ and $\pi_{d^i}(\mathbf{X})$ are bounded below by $\epsilon\leq\hat{\pi}^{k}_{d^i}(\mathbf{X}),\pi_{d^i}(\mathbf{X})\leq 1-\epsilon$, we have
	\begin{equation*}
		{\small
			\begin{aligned}
				&\mathbb{E}[\|H_{2}(D_{s},\mathbf{X}_{s},Z_{s})\|^{2}]\\&=\mathbb{E}\bigg[\bigg\|\frac{\mathbf{1}_{\{D=d^{i}\}}\{\hat{\pi}^{k}_{d^i}(\mathbf{X})m_{d^i}(\mathbf{X})-\pi_{d^i}(\mathbf{X})\tilde{m}_{d^i}^k(\mathbf{X})\}}{\hat{\pi}^{k}_{d^i}(\mathbf{X})\pi_{d^i}(\mathbf{X})}\bigg\|^{2}\bigg]\\
				\lesssim&\mathbb{E}\bigg[\big\|\hat{\pi}^{k}_{d^i}(\mathbf{X})m_{d^i}(\mathbf{X})-\pi_{d^i}(\mathbf{X})m_{d^i}(\mathbf{X})\big\|^{2}\bigg]+\mathbb{E}\bigg[\big\|\pi_{d^i}(\mathbf{X})m_{d^i}(\mathbf{X})-\pi_{d^i}(\mathbf{X})\tilde{m}_{d^i}^k(\mathbf{X})\big\|^{2}\bigg]\\
				\lesssim&\mathbb{E}[|\hat{\pi}^{k}_{d^i}(\mathbf{X})-\pi_{d^i}(\mathbf{X})|^{2}\|m_{d^i}(\mathbf{X})\|^{2}]+\mathbb{E}[\|m_{d^i}(\mathbf{X})-\tilde{m}_{d^i}^k(\mathbf{X})\|^{2}]\\
				\leq&(\mathbb{E}[|\hat{\pi}^{k}_{d^i}(\mathbf{X})-\pi_{d^i}(\mathbf{X})|^{4}])^{\frac{1}{2}}(\mathbb{E}[\|m_{d^i}(\mathbf{X})\|^{4}])^{\frac{1}{2}}+\mathbb{E}[\|m_{d^i}(\mathbf{X})-\tilde{m}_{d^i}^k(\mathbf{X})\|^{2}]\lesssim\rho_{\pi}^{2}+\rho_{m}^{2}.
			\end{aligned}
		}
	\end{equation*}
	Besides, we can bound $\mathbb{E}[\|H_{3}(D_{s},\mathbf{X}_{s},Z_{s})\|^{2}]$ similarly. Indeed, we have
	\begin{equation*}
		{\small
			\begin{aligned}
				&\mathbb{E}[\|H_{3}(D_{s},\mathbf{X}_{s},Z_{s})\|^{2}]=\mathbb{E}[\|\tilde{m}_{d^i}^k(\mathbf{X})-m_{d^i}(\mathbf{X})\|^{2}]\leq\rho_{m}^{2}.
			\end{aligned}
		}
	\end{equation*}
	
	As a result, we have 
	\begin{equation*}
		{\small
			\begin{aligned}
				\text{I}_{1}\lesssim\frac{1}{N_{k}}(\rho_{\pi}^{2}+\rho_{m}^{2})=\frac{N}{N_{k}}\frac{1}{N}(\rho_{\pi}^{2}+\rho_{m}^{2})\lesssim\frac{1}{N}(\rho_{\pi}^{2}+\rho_{m}^{2}).
			\end{aligned}
		}
	\end{equation*} 
	We now consider the quantity $\text{I}_{2}$. Note that
	\begin{equation*}
		{\small
			\begin{aligned}
				H(D,\mathbf{X},Z)&=\bigg[\frac{\mathbf{1}_{\{D=d^{i}\}}(Z-\tilde{m}_{d^i}^k(\mathbf{X}))}{\hat{\pi}^{k}_{d^i}(\mathbf{X})}+\tilde{m}_{d^i}^k(\mathbf{X})-\frac{\mathbf{1}_{\{D=d^{i}\}}(Z-m_{i}(\mathbf{X}))}{\pi_{d^i}(\mathbf{X})}-m_{d^i}(\mathbf{X})\bigg]\\
				&=G(D,\mathbf{X},Z)-\frac{\mathbf{1}_{\{D=d^{i}\}}(Z-m_{d^i}(\mathbf{X}))}{\pi_{d^i}(\mathbf{X})}
			\end{aligned}
		}
	\end{equation*}
	and $\mathbb{E}[H(D,\mathbf{X},Z)]=\mathbb{E}[G(D,\mathbf{X},Z)]$. Abbreviating the notation that $G(D,\mathbf{X},Z)=G$, since the $s^{\text{th}}$-unit and the $\bar{s}^{\text{th}}$-unit are independent of each other, we have
	\begin{equation*}
		{\small
			\begin{aligned}
				&\mathbb{E}[\langle H_{s}-\mathbb{E}[H_{s}],H_{\bar{s}}-\mathbb{E}[H_{\bar{s}}]\rangle]=\mathbb{E}[\langle G_{s}-\mathbb{E}[G_{s}],G_{\bar{s}}-\mathbb{E}[G_{\bar{s}}]\rangle]\\
				=&\mathbb{E}[\langle G_{s},G_{\bar{s}}\rangle]-\langle \mathbb{E}[G_{s}],\mathbb{E}[G_{\bar{s}}]\rangle\lesssim\|\mathbb{E}[G_{s}]\|\times\|\mathbb{E}[G_{\bar{s}}]\|.
			\end{aligned}
		}
	\end{equation*}
	Note that
	\begin{equation*}
		{\small
			\begin{aligned}
				&\mathbb{E}[G_{s}]=\mathbb{E}\bigg[\frac{\mathbf{1}_{\{D=d^{i}\}}\{Z-\tilde{m}_{d^i}^k(\mathbf{X})\}}{\hat{\pi}^{k}_{d^i}(\mathbf{X})}+\tilde{m}_{d^i}^k(\mathbf{X})-m_{d^i}(\mathbf{X})\bigg]\\
				=&\mathbb{E}\bigg[\mathbb{E}\bigg[\frac{\mathbf{1}_{\{D=d^{i}\}}\{Z-\tilde{m}_{d^i}^k(\mathbf{X})\}}{\hat{\pi}^{k}_{d^i}(\mathbf{X})}|\mathbf{X}\bigg]\bigg]+\mathbb{E}[\tilde{m}_{d^i}^k(\mathbf{X})-m_{d^i}(\mathbf{X})]\\
				=&\mathbb{E}\bigg[\frac{\mathbb{E}[\mathbf{1}_{\{D=d^{i}\}}|\mathbf{X}]}{\hat{\pi}^{k}_{d^i}(\mathbf{X})}\{m_{d^i}(\mathbf{X})-\tilde{m}_{d^i}^k(\mathbf{X})\}\bigg]+\mathbb{E}[\tilde{m}_{d^i}^k(\mathbf{X})-m_{d^i}(\mathbf{X})]\\
				=&\mathbb{E}\bigg[\frac{\pi_{d^i}(\mathbf{X})-\pi^{k}_{d^i}(\mathbf{X})}{\hat{\pi}^{k}_{d^i}(\mathbf{X})}\{m_{d^i}(\mathbf{X})-\tilde{m}_{d^i}^k(\mathbf{X})\}\bigg]\\
				\lesssim&\mathbb{E}[\{\pi_{d^i}(\mathbf{X})-\pi^{k}_{d^i}(\mathbf{X})\}\{m_{d^i}(\mathbf{X})-\tilde{m}_{d^i}^k(\mathbf{X})\}].
			\end{aligned}
		}
	\end{equation*}
	Hence, we have
	\begin{equation*}
		{\small
			\begin{aligned}
				\|\mathbb{E}[G_{s}]\|&\lesssim\mathbb{E}[|\pi_{d^i}(\mathbf{X})-\pi^{k}_{d^i}(\mathbf{X})|\times\|m_{d^i}(\mathbf{X})-\tilde{m}_{d^i}^k(\mathbf{X})\|]\\
				&\leq(\mathbb{E}[|\pi_{d^i}(\mathbf{X})-\pi^{k}_{d^i}(\mathbf{X})|^{2}])^{\frac{1}{2}}(\mathbb{E}[\|m_{d^i}(\mathbf{X})-\tilde{m}_{d^i}^k(\mathbf{X})\|^{2}])^{\frac{1}{2}}\\
				&\leq(\mathbb{E}[|\pi_{d^i}(\mathbf{X})-\pi^{k}_{d^i}(\mathbf{X})|^{4}])^{\frac{1}{4}}\rho_{m}=\rho_{\pi}\rho_{m}.
			\end{aligned}
		}
	\end{equation*}
	Hence, we have $C_{s\bar{s}}\lesssim\rho_{\pi}^{2}\rho_{m}^{2}$ and $\text{I}_{2}\lesssim\big(1-\frac{1}{N_{k}}\big)\rho_{\pi}^{2}\rho_{m}^{2}$. As a result, we can show that
	\begin{equation*}
		\begin{aligned}
			\mathbb{E}[\|\text{I}\|^{2}]=O(N^{-1}\rho_{\pi}^{2}+N^{-1}\rho_{m}^{2}+\rho_{\pi}^{2}\rho_{m}^{2}).
		\end{aligned}
	\end{equation*}
	Thus, we have $\text{I}=O_{P}(N^{-\frac{1}{2}}\rho_{\pi}+N^{-\frac{1}{2}}\rho_{m}+\rho_{\pi}\rho_{m})$.

	\noindent\underline{Boundness of II:} Note that the quantity II does not involve any estimation of nuisance functions based on the observed dataset (the function $\pi_{d^i}$) and $m_{d^i}$ are the limits of some estimated sequences). It is equivalent to the following problem:
	
	\begin{claim*}
		Given a random quantity $W$ such that $\mathbb{E}[W]<\infty$, the sample averaging of $W$ (i.e., $\frac{1}{N}\underset{s=1}{\overset{N}{\sum}}W_{s}$) converges to $\mathbb{E}[W]$ in probability, or $ \frac{1}{N}\underset{s=1}{\overset{N}{\sum}}W_{s}-\mathbb{E}[W]=O_{P}\big(\frac{1}{\sqrt{N}}\big)$.
	\end{claim*}
	\noindent Using the Claim and Assumption \ref{ass:assumption6}, we have
	\begin{equation*}
		{\small
			\begin{aligned}
				&\mathbb{P}\bigg\{\bigg|\frac{\text{II}}{\frac{1}{\sqrt{N}}}\bigg|\geq M\bigg\}=\mathbb{P}\bigg\{\bigg|\frac{(\mathbb{P}_{N_{k}}-\mathbb{E}_{N_{k}})\bigg[\frac{\mathbf{1}_{\{D=d^{i}\}}(Z-m_{d^i}(\mathbf{X}))}{\pi_{d^i}(\mathbf{X})}+m_{d^i}(\mathbf{X})\bigg]}{\frac{1}{\sqrt{N}}}\bigg|\geq M\bigg\}\\
				\leq&\frac{N}{N_{k}}\frac{\mathbb{E}\bigg[\bigg(\frac{\mathbf{1}_{\{D=d^{i}\}}(Z-m_{d^i}(\mathbf{X}))}{\pi_{d^i}(\mathbf{X})}+m_{d^i}(\mathbf{X})\bigg)^{2}\bigg]}{M^{2}}.
			\end{aligned}
		}
	\end{equation*}
	We can choose sufficiently large $M$ to make {\small $\mathbb{P}\bigg\{\bigg|\frac{\text{II}}{\frac{1}{\sqrt{N}}}\bigg|\geq M\bigg\}<\epsilon$}. Hence, we have $\text{II}=O_{P}(N^{-\frac{1}{2}})$.
	
	\noindent\underline{Boundness of III:} For simplicity, we denote
	\begin{equation*}
		{\small
			\begin{aligned}
				\mathcal{A}=\mathbb{E}_{N_{k}}\bigg[\frac{(\tilde{m}_{d^i}^k(\mathbf{X})-m_{d^i}(\mathbf{X}))(\hat{\pi}^{k}_{d^i}(\mathbf{X})-\mathbf{1}_{\{D=d^{i}\}})}{\hat{\pi}^{k}_{d^i}(\mathbf{X})}\bigg].
			\end{aligned}
		}
	\end{equation*}
	
	We consider the quantity $\mathbb{E}[\|\mathcal{A}\|]$. Since $\mathcal{A}$ is an expectation already, we have $\mathbb{E}[\|\mathcal{A}\|]=\|\mathcal{A}\|$. Further, we can simplify $\|\mathcal{A}\|$ as follows:
	\begin{equation*}
		{\small
			\begin{aligned}
				\|\mathcal{A}\|&=\bigg\|\mathbb{E}_{N_{k}}\bigg[\frac{(\tilde{m}_{d^i}^k(\mathbf{X})-m_{d^i}(\mathbf{X}))(\hat{\pi}^{k}_{d^i}(\mathbf{X})-\mathbf{1}_{\{D=d^{i}\}})}{\hat{\pi}^{k}_{d^i}(\mathbf{X})}\bigg]\bigg\|\\
				&=\bigg\|\mathbb{E}\bigg[\frac{(\tilde{m}_{d^i}^k(\mathbf{X})-m_{d^i}(\mathbf{X}))(\hat{\pi}^{k}_{d^i}(\mathbf{X})-\mathbf{1}_{\{D=d^{i}\}})}{\hat{\pi}^{k}_{d^i}(\mathbf{X})}\bigg]\bigg\|\\
				&=\bigg\|\mathbb{E}\bigg[\mathbb{E}\bigg[\frac{(\tilde{m}_{d^i}^k(\mathbf{X})-m_{d^i}(\mathbf{X}))(\hat{\pi}^{k}_{d^i}(\mathbf{X})-\mathbf{1}_{\{D=d^{i}\}})}{\hat{\pi}^{k}_{d^i}(\mathbf{X})}|\mathbf{X}\bigg]\bigg]\bigg\|\\
				&=\bigg\|\mathbb{E}\bigg[\frac{(\tilde{m}_{d^i}^k(\mathbf{X})-m_{d^i}(\mathbf{X}))}{\hat{\pi}^{k}_{d^i}(\mathbf{X})}(\hat{\pi}^{k}_{d^i}(\mathbf{X})-\mathbb{E}[\mathbf{1}_{\{D=d^{i}\}}|\mathbf{X}])\bigg]\bigg\|\\
				&=\bigg\|\mathbb{E}\bigg[\frac{(\tilde{m}_{d^i}^k(\mathbf{X})-m_{d^i}(\mathbf{X}))}{\hat{\pi}^{k}_{d^i}(\mathbf{X})}(\hat{\pi}^{k}_{d^i}(\mathbf{X})-\pi_{d^i}(\mathbf{X}))\bigg]\bigg\|\\
				&\lesssim\|\mathbb{E}[(\tilde{m}_{d^i}^k(\mathbf{X})-m_{d^i}(\mathbf{X}))(\hat{\pi}^{k}_{d^i}(\mathbf{X})-\pi_{d^i}(\mathbf{X}))]\|\\
				&\leq\mathbb{E}[|\hat{\pi}^{k}_{d^i}(\mathbf{X})-\pi_{d^i}(\mathbf{X})|\|(\tilde{m}_{d^i}^k(\mathbf{X})-m_{d^i}(\mathbf{X}))\|]\\
				&\leq\big(\mathbb{E}[|\hat{\pi}^{k}_{d^i}(\mathbf{X})-\pi_{d^i}(\mathbf{X})|^{2}]\big)^{\frac{1}{2}}\big(\mathbb{E}[\|(\tilde{m}_{d^i}^k(\mathbf{X})-m_{d^i}(\mathbf{X}))\|^{2}]\big)^{\frac{1}{2}}\\
				&\leq\big(\mathbb{E}[|\hat{\pi}^{k}_{d^i}(\mathbf{X})-\pi_{d^i}(\mathbf{X})|^{4}]\big)^{\frac{1}{4}}\big(\mathbb{E}[\|(\tilde{m}_{d^i}^k(\mathbf{X})-m_{d^i}(\mathbf{X}))\|^{2}]\big)^{\frac{1}{2}}\\ &\leq\rho_{\pi}\rho_{m}.
			\end{aligned}
		}
	\end{equation*}
	\noindent\underline{Boundness of IV:} Let $\mathcal{A}=\mathbb{P}_{N_{k}}\bigg\{\bigg(1-\frac{\mathbf{1}_{\{D=d^{i}\}}}{\hat{\pi}^{k}_{d^i}(\mathbf{X})}\bigg)D_{d^i}^k(\mathbf{X})\bigg\}$.
	
	Consider $\|\mathcal{A}\|^{2}$. We have
	\begin{equation*}
		{\small
			\begin{aligned}
				&\|\mathcal{A}\|^{2}=\bigg\|\mathbb{P}_{N_{k}}\bigg\{\bigg(1-\frac{\mathbf{1}_{\{D=d^{i}\}}}{\hat{\pi}^{k}_{d^i}(\mathbf{X})}\bigg)D_{d^i}^k(\mathbf{X})\bigg\}\bigg\|^{2}\\
				=&\underbrace{\frac{1}{N_{k}^{2}}\underset{s\in\mathcal{D}_{k}}{\sum}\bigg\|\bigg(1-\frac{\mathbf{1}_{\{D_{s}=d^{i}\}}}{\hat{\pi}^{k}_{d^i}(\mathbf{X}_{s})}\bigg)D_{d^i}^k(\mathbf{X}_{s})\bigg\|^{2}}_{\text{IV}_{1}}
				+\underbrace{\frac{1}{N_{k}^{2}}\underset{\substack{s,\bar{s}\in\mathcal{D}^{k}\\s\ne\bar{s}}}{\sum}\langle(1-\frac{\mathbf{1}_{\{D_{s}=d^{i}\}}}{\hat{\pi}^{k}_{d^i}(\mathbf{X}_{s})})D_{i,k}(\mathbf{X}_{s}),(1-\frac{\mathbf{1}_{\{D_{\bar{s}}=d^{i}\}}}{\hat{\pi}^{k}_{d^i}(\mathbf{X}_{\bar{s}})})D_{d^i}^k(\mathbf{X}_{\bar{s}})\rangle}_{\text{IV}_{2}}.
			\end{aligned}
		}
	\end{equation*}
	Consider $\text{IV}_{1}$ first. Using Assumption \ref{ass:assumption3}, we see that $\text{IV}_{1}\leq\frac{c}{N_{k}^{2}}\underset{s\in\mathcal{D}^k}{\sum}\|D_{d^i}^k(\mathbf{X}_{s})\|^{2}$ for some constant $c$. Note that, for any $\delta>0$, we have
	
	\begin{equation*}
		\begin{aligned}
			&\mathbb{P}\bigg\{\frac{1}{N_{k}}\underset{s\in\mathcal{D}^{k}}{\sum}\|D_{d^i}^k(\mathbf{X}_{s})\|^{2}\geq\frac{\vvvert \hat{m}_{d^i}^k-\tilde{m}_{d^i}^k\vvvert^{2}}{\delta}\bigg\}\\
			\leq&\frac{\delta\mathbb{E}\bigg[\frac{1}{N_{k}}\underset{s\in\mathcal{D}^{k}}{\sum}\|D_{d^i}^k(\mathbf{X}_{s})\|^{2}\bigg]}{\vvvert \hat{m}_{d^i}^k-\tilde{m}_{d^i}^k\vvvert^{2}}=\frac{\delta\mathbb{E}[\|D_{d^i}^k(\mathbf{X})\|^{2}]}{\vvvert \hat{m}_{d^i}^k-\tilde{m}_{d^i}^k\vvvert^{2}}=\delta.
		\end{aligned}
	\end{equation*}
	Indeed, the inequality follows from Markov inequality. The last equality follows from the Definition of $\vvvert \cdot\vvvert^{2}$. According to the definition, we have $\mathbb{E}[\|D_{d^i}^k(\mathbf{X})\|^{2}]=\vvvert \hat{m}_{d^i}^k-\tilde{m}_{d^i}^k\vvvert^{2}$. It means that $\frac{1}{N_{k}}\underset{s\in\mathcal{D}^{k}}{\sum}\|D_{d^i}^k(\mathbf{X}_{s})\|^{2}=O_{P}\big(\vvvert \hat{m}_{d^i}^k-\tilde{m}_{d^i}^k\vvvert^{2}\big)$. Hence, we note that $\text{IV}_{1}=\frac{1}{N_{k}}\times\frac{1}{N_{k}}\underset{s\in\mathcal{D}^{k}}{\sum}\|D_{d^i}^k(\mathbf{X}_{s})\|^{2}=\frac{N}{N_{k}}\times\frac{1}{N}\times O_{P}\big(\vvvert\hat{m}_{d^i}^k-\tilde{m}_{d^i}^k\vvvert^{2}\big)$. Using Assumptions 3 and 4, we have $\text{IV}_{1}=O_{P}(N^{-2}+N^{-1}\alpha_{N}^{2}+N^{-1}\nu_{N}^{2})$.
	
	\noindent Next, we consider $\text{IV}_{2}$. Let 
	\begin{equation*}
		\begin{aligned}
			\mathcal{A}&=(\mathbb{E}[\| D_{d^i}^k(\mathbf{X}_{s})\|^{4}])^{\frac{1}{4}}(\mathbb{E}[\|D_{d^i}^k(\mathbf{X}_{\bar{s}})\|^{4}])^{\frac{1}{4}}\times\\
			&\quad(\mathbb{E}[(\hat{\pi}_{d^i}^{k}(\mathbf{X}_{s})-\pi_{d^i}(\mathbf{X}_{s}))^{4}])^{\frac{1}{4}}\mathbb{E}[(\hat{\pi}_{d^i}^{k}(\mathbf{X}_{\bar{s}})-\pi_{d^i}(\mathbf{X}_{\bar{s}}))^{4}])^{\frac{1}{4}}.
		\end{aligned}
	\end{equation*}
	Note that, for any $\delta>0$, we have
	\begin{equation*}
		{\small
			\begin{aligned}
				&\mathbb{P}\bigg\{\text{IV}_{2}\geq\frac{\mathcal{A}}{\delta}\bigg\} \\\leq&\frac{\delta\frac{1}{N_{k}^{2}}\underset{\substack{s,\bar{s}\in\mathcal{D}^{k}\\s\neq \bar{s}}}{\sum}\mathbb{E}\bigg[\langle(1-\frac{\mathbf{1}_{\{D_{s}=d^{i}\}}}{\hat{\pi}_{d^i}^{k}(\mathbf{X}_{s})})D_{d^i}^k(\mathbf{X}_{s}),(1-\frac{\mathbf{1}_{\{D_{\bar{s}}=d^{i}\}}}{\hat{\pi}^{k}_{d^i}(\mathbf{X}_{\bar{s}})})D_{d^i}^k(\mathbf{X}_{\bar{s}})\rangle\bigg]}{\mathcal{A}}\\
				\overset{\star}{\leq}&\frac{\delta\frac{1}{N_{k}^{2}}\underset{\substack{s,\bar{s}\in\mathcal{D}^{k}\\s\neq\bar{s}}}{\sum}\mathcal{A}}{\mathcal{A}}=\frac{\delta N_{k}(N_{k}-1)}{N_{k}^2}=\delta\bigg(1-\frac{1}{N_{k}}\bigg)\leq \delta.
			\end{aligned}
		}
	\end{equation*}
	Here, $\overset{\star}{\leq}$ is due to the upper bound of the quantity $\mathbb{E}\big[\langle(1-\frac{\mathbf{1}_{\{D_{s}=d^{i}\}}}{\hat{\pi}_{d^i}^{k}(\mathbf{X}_{s})})D_{d^i}^k(\mathbf{X}_{s}),(1-\frac{\mathbf{1}_{\{D_{\bar{s}}=d^{i}\}}}{\hat{\pi}^{k}_{d^i}(\mathbf{X}_{\bar{s}})})D_{d^i}^k(\mathbf{X}_{\bar{s}})\rangle\big]$. Indeed, using the fact that the unit $s$ and the unit $\bar{s}$ are independent of each other, we have
	\begin{equation*}
		{\small
			\begin{aligned}
				&\mathbb{E}\bigg[\langle(1-\frac{\mathbf{1}_{\{D_{s}=d^{i}\}}}{\hat{\pi}_{d^i}^{k}(\mathbf{X}_{s})})D_{d^i}^k(\mathbf{X}_{s}),(1-\frac{\mathbf{1}_{\{D_{\bar{s}}=d^{i}\}}}{\hat{\pi}^{k}_{d^i}(\mathbf{X}_{\bar{s}})})D_{d^i}^k(\mathbf{X}_{\bar{s}})\rangle\bigg]\\
				=&\mathbb{E}\bigg[(1-\frac{\mathbf{1}_{\{D_{s}=d^{i}\}}}{\hat{\pi}_{d^i}^{k}(\mathbf{X}_{s})})(1-\frac{\mathbf{1}_{\{D_{\bar{s}}=d^{i}\}}}{\hat{\pi}^{k}_{d^i}(\mathbf{X}_{\bar{s}})})\langle D_{d^i}^k(\mathbf{X}_{s}),D_{d^i}^k(\mathbf{X}_{\bar{s}})\rangle\bigg]\\
				=&\mathbb{E}\bigg[\langle D_{d^i}^k(\mathbf{X}_{s}),D_{d^i}^k(\mathbf{X}_{\bar{s}})\rangle\mathbb{E}\bigg[(1-\frac{\mathbf{1}_{\{D_{s}=d^{i}\}}}{\hat{\pi}_{d^i}^{k}(\mathbf{X}_{s})})(1-\frac{\mathbf{1}_{\{D_{\bar{s}}=d^{i}\}}}{\hat{\pi}^{k}_{d^i}(\mathbf{X}_{\bar{s}})})|\mathbf{X}\bigg]\bigg]\\
				=&\mathbb{E}\bigg[\langle D_{d^i}^k(\mathbf{X}_{s}),D_{d^i}^k(\mathbf{X}_{\bar{s}})\rangle\mathbb{E}\bigg[(1-\frac{\mathbf{1}_{\{D_{s}=d^{i}\}}}{\hat{\pi}_{d^i}^{k}(\mathbf{X}_{s})})|\mathbf{X}\bigg]\mathbb{E}\bigg[(1-\frac{\mathbf{1}_{\{D_{\bar{s}}=d^{i}\}}}{\hat{\pi}^{k}_{d^i}(\mathbf{X}_{\bar{s}})})|\mathbf{X}\bigg]\bigg].
			\end{aligned}
		}
	\end{equation*} 
	Now, using Assumption \ref{ass:assumption3}, we can further have
	\begin{equation*}
		{\small
			\begin{aligned}
				&\mathbb{E}\bigg[\langle(1-\frac{\mathbf{1}_{\{D_{s}=d^{i}\}}}{\hat{\pi}_{d^i}^{k}(\mathbf{X}_{s})})D_{d^i}^k(\mathbf{X}_{s}),(1-\frac{\mathbf{1}_{\{D_{\bar{s}}=d^{i}\}}}{\hat{\pi}^{k}_{d^i}(\mathbf{X}_{\bar{s}})})D_{d^i}^k(\mathbf{X}_{\bar{s}})\rangle\bigg]\\
				\lesssim&\mathbb{E}[\langle D_{d^i}^k(\mathbf{X}_{s}),D_{d^i}^k(\mathbf{X}_{\bar{s}})\rangle(\hat{\pi}_{d^i}^{k}(\mathbf{X}_{s})-\pi_{d^i}(\mathbf{X}_{s}))(\hat{\pi}_{d^i}^{k}(\mathbf{X}_{\bar{s}})-\pi_{d^i}(\mathbf{X}_{\bar{s}}))]\\
				\leq&(\mathbb{E}[\langle D_{d^i}^k(\mathbf{X}_{s}),D_{d^i}^k(\mathbf{X}_{\bar{s}})\rangle^{2}])^{\frac{1}{2}}\times\\
				&\quad(\mathbb{E}[(\hat{\pi}_{d^i}^{k}(\mathbf{X}_{s})-\pi_{d^i}(\mathbf{X}_{s}))^{2}(\hat{\pi}_{d^i}^{k}(\mathbf{X}_{\bar{s}})-\pi_{d^i}(\mathbf{X}_{\bar{s}}))^{2}])^{\frac{1}{2}}\\
				\leq&(\mathbb{E}[\| D_{d^i}^k(\mathbf{X}_{s})\|^{2}\|D_{d^i}^k(\mathbf{X}_{\bar{s}})\|^{2}])^{\frac{1}{2}}\times\\
				&\quad(\mathbb{E}[(\hat{\pi}_{d^i}^{k}(\mathbf{X}_{s})-\pi_{d^i}(\mathbf{X}_{s}))^{4}])^{\frac{1}{4}}\mathbb{E}[(\hat{\pi}_{d^i}^{k}(\mathbf{X}_{\bar{s}})-\pi_{d^i}(\mathbf{X}_{\bar{s}}))^{4}])^{\frac{1}{4}}\\
				\leq&(\mathbb{E}[\| D_{d^i}^k(\mathbf{X}_{s})\|^{4}])^{\frac{1}{4}}(\mathbb{E}[\|D_{d^i}^k(\mathbf{X}_{\bar{s}})\|^{4}])^{\frac{1}{4}}\\&(\mathbb{E}[(\hat{\pi}_{d^i}^{k}(\mathbf{X}_{s})-\pi_{d^i}(\mathbf{X}_{s}))^{4}])^{\frac{1}{4}}\mathbb{E}[(\hat{\pi}_{d^i}^{k}(\mathbf{X}_{\bar{s}})-\pi_{d^i}(\mathbf{X}_{\bar{s}}))^{4}])^{\frac{1}{4}}\\
				=&(\mathbb{E}[\| D_{d^i}^k(\mathbf{X})\|^{4}])^{\frac{1}{2}}(\mathbb{E}[(\hat{\pi}_{d^i}^{k}(\mathbf{X})-\pi_{d^i}(\mathbf{X}))^{4}])^{\frac{1}{2}}\\
				=&O_{P}\big(\rho_{\pi}^{2}\vvvert \hat{m}_{d^i}^k-\tilde{m}_{d^i}^k\vvvert^{2}\big).
			\end{aligned}
		}
	\end{equation*}
	Hence, we can conclude that $\text{IV}_{2}=O_{P}(\rho_{\pi}^{2}N^{-1}+\rho_{\pi}^{2}\alpha_{N}^{2}+\rho_{\pi}^{2}\nu^{2})$. Consequently, we obtain that $\|\text{IV}\|_{\lambda}^{2}=O_{P}(N^{-2}+N^{-1}\alpha_{N}^{2}+N^{-1}\nu_{N}^{2}+\rho_{\pi}^{2}N^{-1}+\rho_{\pi}^{2}\alpha_{N}^{2}+\rho_{\pi}^{2}\nu^{2})$, implying that $\text{IV}=O_{P}(N^{-1}+N^{-\frac{1}{2}}\alpha_{N}+N^{-\frac{1}{2}}\nu_{N}+\rho_{\pi}N^{-\frac{1}{2}}+\rho_{\pi}\alpha_{N}+\rho_{\pi}\nu)$.
	
	\noindent\underline{Boundness of V:} Observe that
	\begin{equation}
		{\small
			\begin{aligned}
				&\mathbb{P}_{N_{k}}\bigg\{\frac{\mathbf{1}_{\{D=d^{i}\}}R}{\hat{\pi}^{k}_{d^i}(\mathbf{X})}\bigg\}\\=&\mathbb{P}_{N_{k}}\bigg\{\frac{\mathbf{1}_{\{D=d^{i}\}}R}{\pi_{d^i}(\mathbf{X})}\bigg\}+\mathbb{P}_{N_{k}}\bigg\{\frac{\mathbf{1}_{\{D=d^{i}\}}R}{\hat{\pi}^{k}_{d^i}(\mathbf{X})}-\frac{\mathbf{1}_{\{D=d^{i}\}}R}{\pi_{d^i}(\mathbf{X})}\bigg\}\\
				=&\mathbb{P}_{N_{k}}\bigg\{\frac{\mathbf{1}_{\{D=d^{i}\}}R}{\pi_{d^i}(\mathbf{X})}\bigg\}+\mathbb{P}_{N_{k}}\bigg\{\frac{\mathbf{1}_{\{D=d^{i}\}}R(\pi_{d^i}(\mathbf{X})-\hat{\pi}^{k}_{d^i}(\mathbf{X}))}{\hat{\pi}^{k}_{d^i}(\mathbf{X})\pi_{d^i}(\mathbf{X})}\bigg\}.\label{eqt:Vterm}
			\end{aligned}
		}
	\end{equation}
	Using the fact that $0<\epsilon\leq \hat{\pi}_{k}^{i}(\mathbf{X}),\pi^{i}(\mathbf{X})\leq 1-\epsilon<1$, we have
	\begin{equation*}
		{\small
			\begin{aligned}
				\text{Eqn. \eqref{eqt:Vterm}}\lesssim&\mathbb{P}_{N_{k}}\bigg\{\frac{\mathbf{1}_{\{D=d^{i}\}}R}{\pi_{d^i}(\mathbf{X})}\bigg\}+\mathbb{P}_{N_{k}}\bigg\{\frac{\mathbf{1}_{\{D=d^{i}\}}R(\pi_{d^i}(\mathbf{X})-\hat{\pi}^{k}_{d^i}(\mathbf{X}))}{\pi_{d^i}(\mathbf{X})}\bigg\}\\
				\leq&\mathbb{P}_{N_{k}}\bigg\{\frac{\mathbf{1}_{\{D=d^{i}\}}R}{\pi_{d^i}(\mathbf{X})}\bigg\}+2\mathbb{P}_{N_{k}}\bigg\{\frac{\mathbf{1}_{\{D=d^{i}\}}R}{\pi_{d^i}(\mathbf{X})}\bigg\}\\&=3\mathbb{P}_{N_{k}}\bigg\{\frac{\mathbf{1}_{\{D=d^{i}\}}R}{\pi_{d^i}(\mathbf{X})}\bigg\}.
			\end{aligned}
		}
	\end{equation*}
	Hence, $\mathbb{P}_{N_{k}}\bigg\{\frac{\mathbf{1}_{\{D=d^{i}\}}R}{\hat{\pi}^{k}_{d^i}(\mathbf{X})}\bigg\}$ is dominated by the quantity $\mathbb{P}_{N_{k}}\bigg\{\frac{\mathbf{1}_{\{D=d^{i}\}}R}{\pi_{d^i}(\mathbf{X})}\bigg\}$. Besides, we can rewrite $\mathbb{P}_{N_{k}}\bigg\{\frac{\mathbf{1}_{\{D=d^{i}\}}R}{\pi_{d^i}(\mathbf{X})}\bigg\}$ such that
	\begin{equation*}
		{\small
			\begin{aligned}
				&\mathbb{P}_{N_{k}}\bigg\{\frac{\mathbf{1}_{\{D=d^{i}\}}R}{\pi_{d^i}(\mathbf{X})}\bigg\}=\mathbb{P}_{N_{k}}\bigg\{\frac{\mathbf{1}_{\{D=d^{i}\}}(\mathcal{L}\hat{\mathcal{Y}}-\mathcal{L}\mathcal{Y})}{\pi_{d^{i}}(\mathbf{X})}\bigg\}.
			\end{aligned}
		}
	\end{equation*}
	
	By Lemma \ref{lemma:2Lemma} and the assumptions of $\alpha_{N}$ and $\nu_{N}$ given in Convergence Assumption \ref{ass:assumption1}, we have $\text{V}=O_{P}(\alpha_{N}+\nu_{N})$.
\end{proof}

\section{Synthetic Experiments}\label{sec:Simulation}
\subsection{Data Generation Process and Dataset Description}
We simulate one synthetic dataset with non-linear causal maps and sample selection bias to test our proposed framework. The data-generating process (DGP) is as follows:

\begin{equation*}
	{\small
		\begin{aligned}
			\mathcal{Y}_s(D_{s})^{-1}&=c+(1-c)(\mathbb{E}[D]+\sqrt{D_s})\times\\
			&\quad\quad\underset{j=1}{\overset{\frac{n}{2}}{\sum}} \frac{\exp(X^{2j-1}_{s}X^{2j}_s)}{\underset{k=1}{\overset{\frac{n}{2}}{\sum}}\exp(X^{2k-1}_{s}X^{2k}_{s})}\mathbf{B}^{-1}(\alpha_j,\beta_j)+\epsilon_{s},\\
			D_s&=\sigma\biggl(\frac{\exp(\gamma_{w}^{T} \mathbf{X}_{s})}{\underset{w=1}{\overset{r}{\sum}}\exp(\gamma_{w}^{T} \mathbf{X}_{s})}\biggl),
		\end{aligned}
	}
\end{equation*}
where $n$ is an even number which indicates the number of covariates, $\mathcal{Y}^{-1}_s$ is the inverse distribution of unit $s$, $\mathbf{X}_{s}$ is the covariates of unit $s$, and $D_{s}$ is the treatment of unit $s$. $\mathbf{B}^{-1}(\alpha,\beta)$ is the inverse cumulative distribution function (CDF) of beta distribution with the shapes' parameters $\alpha$ and $\beta$. $\sigma$ is a function that maps the features to treatment $D$ such that $D$ takes five treatment levels $\{d^1, d^2, d^3, d^4, d^5\}$. $c$ is the constant that controls the strength of the causal relationship between $D$ and $\mathcal{Y}^{-1}$. $\epsilon$ is the noise that follows $N(0,0.05)$. In our experiments, we set $n=10$. We assume that $X^1, X^2 \sim\mathcal{N}(-2,1), X^3, X^4 \sim\mathcal{N}(-1,1), X^5, X^6 \sim\mathcal{N}(0,1), X^7, X^8 \sim\mathcal{N}(1,1), X^9, X^{10} \sim\mathcal{N}(2,1)$. $5$ inverse beta CDFs are needed, and we set each beta distribution with different parameters to ensure the complexity of the distribution function. For each unit $s$, $100$ observations are sampled from inverse CDF using the inverse transform sampling method. In one experiment, $5,000$ instances are generated.
\subsection{Training Details}
The hyperparameters are tuned using the random search for both models, and we set the hyperparameters as follows: learning rate: 0.003, batch size: 128, number of epochs: 150, dropout: 0.1, weight decay: 0.001. We use Adam as the optimizer. The adaptive learning rate is used for training, and if the test accuracy does not decrease for 10 epochs, the learning rate will decrease by half.


\begin{table*}[htb]
	\centering
	\caption{The statistical description of training features. \label{tab:statistical decriptions}}
	\begin{tabular}{@{}cccccccc@{}}
		\toprule
		& mean        & std      & min & 25\%    & 50\%   & 75\%    & max     \\ \midrule
		Age           & 35.211    & 7.272    & 21  & 30      & 34     & 39      & 60      \\
		Gender        & 0.598     & 0.490    & 0   & 0       & 1      & 1       & 1       \\
		Num of orders & 209.801   & 108.096  & 78  & 147     & 180    & 232     & 2187    \\
		Default rate  & 0.357     & 0.479    & 0   & 0       & 0      & 1       & 1       \\
		Credit Line   & 10190.840 & 4516.421 & 0   & 7342    & 10000  & 11707   & 51792   \\
		$Q=10\%$      & 24.467    & 21.315   & 0   & 6       & 23.61  & 35.934  & 182.98  \\
		$Q=20\%$      & 41.589    & 33.711   & 0   & 14.88   & 39     & 59.97   & 779     \\
		$Q=30\%$      & 60.316    & 76.464   & 0   & 27.97   & 57.4   & 88.952  & 3999    \\
		$Q=40\%$      & 80.459    & 94.739   & 0   & 43.979  & 80.722 & 104.151 & 4799    \\
		$Q=50\%$      & 105.422   & 120.121  & 0   & 63.995  & 99.93  & 126.665 & 4998.99 \\
		$Q=60\%$      & 139.409   & 161.367  & 0   & 90.767  & 119.24 & 162.587 & 5298    \\
		$Q=70\%$      & 195.481   & 258.942  & 0   & 105.926 & 155    & 215.935 & 7364    \\
		$Q=80\%$      & 298.639   & 415.584  & 0   & 147.824 & 211.64 & 315.036 & 8380    \\
		$Q=90\%$      & 588.585   & 827.340  & 0   & 237.13  & 365.12 & 599.442 & 15320   \\ \bottomrule
	\end{tabular}
\end{table*}

\section{Empirical Experiments}
\label{sec:Empirical Experiments}
Our data is collected from a large E-commerce platform that introduced the ``Buy now, pay later'' (BNPL) credit service to boost the spending of consumers. The objective of our study was to examine the impact of changes in credit lines on spending distribution. We collect data from 4,043 users on the platform over a period of 24 months, from January 2018 to December 2019. The data included demographic information such as gender, age, and location, as well as shopping behaviors such as the amount paid for each order, the total number of orders, and financial information such as the presence of default records, the total number of loans, and the credit line assigned by the platform. To eliminate the impact of the two promotion seasons in June and November, we selected July, August, September, and October of 2019 as our target research period. All the paid amounts of each order by each user during this period constitute a spending distribution for each user.

The statistical descriptions of the aforementioned features are presented in Table \ref{tab:statistical decriptions}. Specifically, the average age of users in our data is 35, the number of males accounts for 59.8\%, and the number of females accounts for 41.2\%. The mean age of the users in our data was found to be 35, with 59.8\% being male and 41.2\% being female. Additionally, the mean number of orders was 210, indicating a relatively high level of activity among the users in our sample. The presence of default records in 35.7\% of the users highlights the need to consider financial stability. Additionally, with respect to the spending distribution, the mean of the distribution, computed from quantiles 0.1 to 0.9, varies from 24.467 to 588.585. The large variance of spending at each quantile highlights the substantial variability of the distribution of consumption among users, demonstrating the diversity of spending patterns among the individuals in the sample.

\begin{table*}[!ht]
	\centering
	\caption{The results of the empirical experiment. \label{tab:detailed res}}
	\begin{tabular}{cccc} 
		\toprule
		\textbf{Quantiles} & \textbf{Low (0-9,000)} & \textbf{Middle (9,000-15,000)} & \textbf{High (\textgreater{}15,000)}  \\ 
		\midrule
		\textbf{5\%}       & 19.9 (19.6, 20.3)      & 21.3 (21.1, 21.4)              & 22.0 (21.2, 22.7)                     \\
		\textbf{10\%}      & 28.0 (27.8, 28.3)      & 29.9 (29.8, 30.0)              & 30.6 (30.0, 31.2)                       \\
		\textbf{15\%}      & 35.4 (35.1, 35.7)      & 37.4 (37.3, 37.6)              & 39.2 (38.6, 39.9)                     \\
		\textbf{20\%}      & 43.6 (43.4, 43.9)      & 47.4 (47.1, 47.8)              & 48.8 (47.9, 49.5)                     \\
		\textbf{25\%}      & 51.2 (50.9, 51.6)      & 57.2 (56.8, 57.5)              & 57.8 (56.9, 58.5)                     \\
		\textbf{30\%}      & 58.7 (58.3, 59.0)      & 65.5 (65.1, 65.9)              & 67.5 (66.5, 68.4)                     \\
		\textbf{35\%}      & 66.6 (66.2, 67.1)      & 75.7 (75.3, 76.1)              & 78.1 (77.0, 79.0)                     \\
		\textbf{40\%}      & 75.2 (74.7, 75.6)      & 86.8 (86.3, 87.2)              & 91.0 (90.0, 92.1)                       \\
		\textbf{45\%}      & 83.9 (83.4, 84.5)      & 99.6 (98.8, 100.5)             & 104.3 (103.3, 105.4)                  \\
		\textbf{50\%}      & 94.9 (94.3, 95.6)      & 115.8 (114.9, 116.9)           & 122.4 (121.1, 123.8)                  \\
		\textbf{55\%}      & 105.6 (104.9, 106.2)   & 131.1 (130.0, 132.2)           & 142.8 (140.6, 145.4)                  \\
		\textbf{60\%}      & 119.0 (118.2, 119.7)   & 150.8 (149.7, 152.0)           & 170.8 (167.4, 174.7)                  \\
		\textbf{65\%}      & 134.8 (133.8, 135.7)   & 174.0 (172.7, 175.3)           & 206.3 (202.3, 210.9)                  \\
		\textbf{70\%}      & 155.1 (153.7, 156.4)   & 207.0 (205.6, 208.5)           & 256.0 (251.8, 261.5)                  \\
		\textbf{75\%}      & 180.7 (178.8, 182.2)   & 259.0 (256.9, 261.1)           & 327.4 (321.5, 334.5)                  \\
		\textbf{80\%}      & 212.9 (210.8, 214.6)   & 325.6 (323.2, 328.3)           & 433.0 (424.4, 442.7)                  \\
		\textbf{85\%}      & 264.1 (260.3, 266.9)   & 434.4 (431.4, 437.7)           & 615.5 (605.7, 628.6)                  \\
		\textbf{90\%}      & 381.0 (374.1, 386.7)   & 654.5 (650.7, 658.4)           & 1020.3 (1003.8, 1036.9)               \\
		\textbf{95\%}      & 801.6 (784.5, 817.2)   & 1217.4 (1211.2, 1225.0)        & 1770.7 (1731.5, 1803.3)              \\ \bottomrule
	\end{tabular}
\end{table*}
The hyperparameters are tuned using random searching for both models. We set the hyperparameters as follows: learning rate: 0.005, batch size: 128, number of epochs: 150, dropout: 0.1, weight decay: 0.001. We use Adam as the optimizer. The adaptive learning rate is used for training, and if the test accuracy does not decrease for 10 epochs, the learning rate will decrease by half. Here, we give the specific values in Table \ref{tab:detailed res}. Generally, the lower quantile stands for some small amount of spending, such as life necessities, while the higher quantiles represent the larger amount of spending, such as luxury goods or services.

Consistent with prior research, our findings indicate a positive correlation between credit lines and spending, highlighting the stimulating impact of credit on consumption. Additionally, we find that such an effect is heterogeneous across different quantiles. Specifically, when the credit lines increase (e.g., from Low to Middle or from Middle to High), the spending on higher quantiles grows significantly (e.g., increase from 801.57 to 1217.40 or from 1217.40 to 1770.66 at quantile 95\%) while the spending on lower quantiles increases relatively slowly (e.g., increase from 19.94 to 21.27 or from 21.27 to 21.96). This suggests that users tend to increase their spending on luxury goods or services when they are able to access credit.

\section{G. Computation Infrastructure}\label{sec:Infrastructure}
All experiments are run on Dell 7920 with Intel(R) Xeon(R) Gold 6250 CPU at 3.90GHz, and a set of NVIDIA Quadro RTX 6000 GP. All models are implemented in Python 3.8. The versions of the main packages of our code are Pytorch 1.8.1+cu102, Sklearn: 0.23.2, Numpy: 1.19.2, Pandas: 1.1.3, Matplotlib:
3.3.2.


\end{document}